%% file: arxiv_main.tex
\documentclass[11pt]{article}

\usepackage{geometry}
\usepackage[utf8]{inputenc}
\usepackage{authblk}

\usepackage{graphicx}
\usepackage{amsmath} 
\usepackage{amssymb} 
\usepackage{amsthm} 
\usepackage{xcolor}
\usepackage{hyperref}

\usepackage{algorithm}
\usepackage{mdframed}

\mdfdefinestyle{figstyle}{ %
  linecolor=black!7, %
  backgroundcolor=black!7, %
  innertopmargin=10pt, %
  innerleftmargin=\textwidth, %
  innerrightmargin=\textwidth, %
  innerbottommargin=5pt %
}

\newtheorem{theorem}{Theorem}[section]

\newtheorem{lemma}[theorem]{Lemma}
\newtheorem{claim}[theorem]{Claim}

\newtheorem{corollary}[theorem]{Corollary}
\newtheorem{definition}[theorem]{Definition}
\newtheorem{observation}{Observation}
\newtheorem{remark}{Remark}

\DeclareMathOperator{\Tr}{Tr}

\newcommand*{\ket}[1]{|#1\rangle}
\newcommand*{\opro}[2]{|#1\rangle\langle#2|}
\newcommand*{\ipro}[2]{\langle #1|#2\rangle}

\newcommand*{\E}{\mathop{\mathbb{E}}}
\newcommand*{\Ora}{\mathcal{O}}
\newcommand*{\Oru}{\mathcal{O}^{f,d}_{\mathrm{unif}}}
\newcommand*{\Z}{\mathbb{Z}}

\newcommand*{\vecf}{\mathcal{F}}
\newcommand*{\nvecf}{\hat{\mathcal{F}}}
\newcommand*{\vecg}{\mathcal{G}}
\newcommand*{\convg}{\mathbf{G}}

\newcommand*{\vecS}{\overline{S}}
\newcommand*{\vecT}{\overline{T}}

\newcommand{\X}{\mathbf{X}}

\newcommand{\Path}{P}
\newcommand{\vecP}{\overrightarrow{P}}

\newcommand*{\A}{\mathcal{A}}
\newcommand*{\B}{\mathcal{B}}

\newcommand*{\F}{\mathbf{F}}
\newcommand*{\Hf}{\mathcal{H}}
\newcommand*{\Hc}{\mathbf{H}}
\newcommand*{\D}{\mathcal{D}}
\newcommand*{\Lang}{\mathcal{L}}

\newcommand*{\SSP}[1]{#1\mbox{-}\class{SSP}}
\newcommand{\Sb}{\mathbf{S}}
\newcommand{\Sf}{\mathbf{SHUF}}

\newcommand{\class}[1]{\mathsf{#1}}

\DeclareMathOperator{\poly}{\mathsf{poly}}
\DeclareMathOperator{\polylog}{\mathsf{polylog}}
\newcommand*{\regI}{\mathbf{I}}
\newcommand*{\regR}{\mathbf{R}}

\title{On the Need for Large Quantum Depth} 

\author[1]{Nai-Hui Chia}
\author[2]{Kai-Min Chung}
\author[3]{Ching-Yi Lai}

\affil[1]{Department of Computer Science, University of Texas at Austin, Texas 78731, USA}
\affil[2]{Institute of Information Science, Academia Sinica, Taipei 11529, Taiwan}
\affil[3]{Institute of Communications Engineering, National Chiao Tung University, Hsinchu 30010, Taiwan}

\date{}

\begin{document}

\maketitle

\begin{abstract}
Near-term quantum computers are likely to have small depths due to short coherence time and noisy gates.  A natural approach to leverage these quantum computers is interleaving them with classical computers. Understanding the capabilities and limits of this hybrid approach is an essential topic in quantum computation. Most notably, the quantum Fourier transform can be implemented by a hybrid of logarithmic-depth quantum circuits and a classical polynomial-time algorithm. Therefore, it seems possible that quantum polylogarithmic depth is as powerful as quantum polynomial depth in the presence of classical computation. 

Indeed, Jozsa conjectured that ``\emph{Any quantum polynomial-time algorithm can be implemented with only $O(\log n)$ quantum depth interspersed with polynomial-time classical computations.}''
This can be formalized as asserting the equivalence of $\class{BQP}$ and ``$\class{BQNC^{BPP}}$''.
On the other hand, Aaronson conjectured that ``\emph{there exists an oracle separation between $\class{BQP}$ and $\class{BPP^{BQNC}}$.}''
$\class{BQNC^{BPP}}$ and $\class{BPP^{BQNC}}$ are two natural and seeming incomparable ways of hybrid classical-quantum computation.

In this work, we manage to prove Aaronson's conjecture and in the meantime disproves Jozsa's conjecture relative to an oracle. In fact, we prove a stronger statement that for any depth parameter $d$, there exists an oracle that separates quantum depth $d$ and $2d+1$ in the presence of classical computation. Thus, our results show that relative to oracles, doubling the quantum circuit depth indeed gives the hybrid model more power, and this cannot be traded by classical computation.       
\end{abstract}

\input{intro.tex}

\input{overview.tex}

\input{definition.tex}

\input{Sora.tex}

\input{O2H_Sora.tex}

\input{QCd.tex}
\input{QCd_BPP.tex}

\input{BPP_QCd_2.tex}

\bibliographystyle{alpha}
\bibliography{qecc} 

\end{document}

%% file: intro.tex
\section{Introduction}

Circuit depth may become an essential consideration when designing algorithms on near-term quantum computers.
Quantum computers with more than 50 qubits have been realized recently by Google~\cite{GOOGLE} and IBM~\cite{IBM}; both the quantity and quality of the qubits are continually improving. Furthermore, Google and NASA recently showed that their quantum computer outperforms the best supercomputers on the task of random circuit sampling~\cite{Arute2019}. However, due to noisy gates and limited coherence time, these quantum computers are only able to operate for a short period. Hence the effective circuit depths of these quantum computers are limited,\footnote{Indeed, the experiments of Google and NASA consider circuits with depth at most 20.} and this seems to be an essential bottleneck for quantum technologies. 


Studies from a theoretical perspective indicate that small-depth quantum computers can demonstrate so-called ``Quantum Supremacy''~\cite{AC17, TD04}, which means solving some computational problems that are intractable for classical computers.  Terhal and DiVincenzo first showed that constant-depth quantum circuits can sample certain distributions that are intractable for classical computers 
under some plausible complexity conjecture~\cite{TD04}, and more recent works showed that this can be based on the conjecture that the polynomial hierarchy (PH) is infinite \cite{AC17}. Aaronson and Chen showed that under a natural average-case hardness assumption, there exists a statistical test such that no polynomial-time classical algorithm can pass it, but a small-depth quantum circuit can~\cite{AC17}. 

On the other hand, it is worth noting that the capability of constant-depth quantum computers is limited. Specifically, consider the standard setting where the composing gate set only includes one- and two-qubit gates. It is obvious that a constant-depth quantum circuit cannot solve any classically intractable decision problem or even some classically easy problems, e.g., computing a parity function. This is because that each output qubit depends on only $O(1)$ input qubits. Although involving unbounded fan-out gates allows a quantum circuit to conduct many operations in small depth, such as parity, mod[q], 
threshold[t], 
arithmetic operations, phase estimations, and the quantum Fourier transform~\cite{HS05}, it seems that unbounded fan-out gates are hard to implement in practice and thus is rarely considered for near-term quantum device.\footnote{Our oracle separation results hold even if the unbounded fan-out gates are allowed.}

A natural idea to exploit the power of small-depth quantum computers is interleaving them with classical computers. Many of the known quantum algorithms require only small depths in the quantum part. Notably, Cleve and Watrous showed that the quantum Fourier transform can be parallelized to have only logarithmic quantum depth~\cite{CW00}, which implies that quantum algorithms for abelian hidden subgroup problems, such as Shor's factoring algorithm, can be implemented with logarithmic quantum depth.  
Therefore, ``quantum polylogarithmic depth is as powerful as quantum polynomial depth in the presence of classical computation'' seems to be a live possibility! 

Aware of this possibility, Jozsa~\cite{Jozsa05} conjectured that 
\begin{itemize}
    \item[] ``\textsl{Any polynomial time quantum algorithm can be implemented with only $O(\log n)$ quantum depth interspersed with polynomial-time classical computations.}''
\end{itemize}
Nevertheless, there  are other opinions in the community. 
It has been conjectured by Aaronson~\cite{Aaronson05,Aaronson10,Aaronson11,Aaronson19} more than a decade ago that
\begin{itemize}
    \item[] ``\emph{There exists an oracle $\mathcal{O}$ relative to which $\class{BQP}\neq \class{BPP^{QNC}}$.}'' 
\end{itemize}
Here, $\class{BPP^{BQNC}}$ is corresponding to one of the hybrid approaches for interleaving classical computers with small-depth quantum circuits. It is worth noting that the models Jozsa and Aaronson considered were related but different, and we will clarify this later in Section~\ref{sec:discussion}. 



In this work, we disprove Jozsa's conjecture and prove Aaronson's conjecture. In fact, we prove a stronger conclusion that relative to oracles, doubling the quantum circuit depth already gives the computational model more power, and this cannot be traded by classical computations. 

\subsection{Main Results}
  We start by defining two hybrid models, which interleave $d$-depth quantum circuits and classical computers. The first scheme, called $d$-depth quantum-classical scheme ($d$-QC scheme), is a generalized model for small-depth measurement-based quantum computers (MBQC). The second scheme, called $d$-depth classical-quantum scheme ($d$-CQ scheme),  characterizes the hybrid quantum and classical computations. Briefly, a $d$-QC scheme is based on a $d$-depth quantum circuit and can access some classical computational resources after each level, while a $d$-CQ scheme is based on a classical computer  with 
access to some $d$-depth quantum circuits. We (informally) define $\class{BQNC_d^{BPP}}$ to be the set of languages decided by $d$-QC schemes, and $\class{BPP^{BQNC_d}}$ to be the set of languages decided by $d$-CQ schemes.  
Also, when we write $\class{BQNC^{BPP}}$ (resp., $\class{BPP^{BQNC}}$), we refer to the union of $\class{BQNC_{\log^k n}^{BPP}}$ (resp., $\class{BPP^{BQNC_{\log^k n}}}$) for constant $k \in \mathbb{N} $.   (These definitions will be formally given in Section~\ref{sec:preliminaries}.)  

Note that Jozsa's conjecture refers to the equivalence problem of BQP and the hybrid model $d$-QC schemes, and Aaronson's conjecture refers to  the separation between BQP and the other hybrid model $d$-CQ schemes. 
As we shall see, we will introduce an oracle problem below, which can be used to show separation results for both conjectures.

Our oracle problem is a variant of \emph{Simon's problem}. Given a function $f:\Z_2^n\rightarrow \Z_2^n$ with the promise that there exists $s\in \Z_2^n$ such that $f(x)=f(x\oplus s)$ for $x\in \Z_2^n$, Simon's problem is to find~$s$~\cite{Simon94}. Such two-to-one function $f$ is called a \emph{Simon function}.  Simon's problem is easy for quantum-polynomial-time (QPT) algorithms but hard for all probabilistic-polynomial-time (PPT) algorithms; however, Simon's problem can be solved by a constant-depth quantum circuit with classical postprocessing. 

To devise a harder problem, we first represent a Simon function as a composition of random one-to-one functions $f_0,\dots,f_{d-1}$ and a  {two-to-one} function $f_d$ such that $f= f_d\circ \cdots\circ f_0$ as shown in Fig.~\ref{fig:oracle_naive}, where $f_j:S_j\rightarrow S_{j+1}$ has domain $S_j$ and range $S_{j+1}$ for $j=0,\dots,d$.   Suppose that $f$ is hidden and only access to these functions $f_i$ are provided. 
If $f_0,\dots,f_d$ must be queried in sequence to evaluate $f$, then a $d$-depth quantum circuit cannot solve this variant of Simon's problem, since it can only make at most $d$ sequential queries (but there are $d+1$ random functions). 

\begin{figure}[h]
    \centering
    \includegraphics[width=0.8\textwidth]{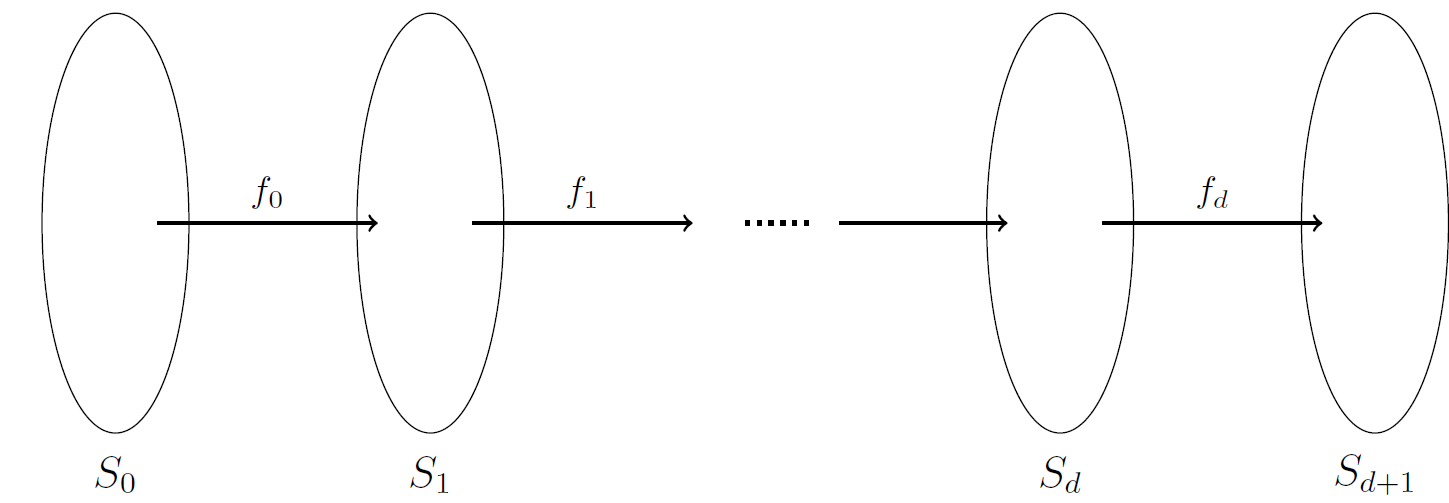} 
    \caption{Composing $d$ random one-to-one function $f_0,\dots,f_{d-1}$ and a two-to-one function $f_d$ such that $f=f_d\circ\cdots\circ f_0$.}
    \label{fig:oracle_naive}
\end{figure}

Nevertheless, a cleverer approach that does not query the functions in sequence might exist. 
To rule out the possibility of such an approach, we further {\bf make it infeasible for  the domain of $f_i$ to be accessed before the $(i+1)$-th parallel queries.}
Specifically, 
$f_j$ is now defined on a larger domain $S_j^{(0)}$ such that $S_j$ is a subset of $S_j^{(0)}$ and $|S_j|/|S_j^{(0)}|$ is chosen to be negligible for $j=0,1,\dots ,d+1$.
 This is illustrated in Fig.~\ref{fig:good_oracle}. 
Therefore, to evaluate $f_i$, one must find the target domain $S_i$ from the larger domain $S_{i}^{(0)}$, and the success probability is negligible by construction.   
\begin{figure}
    \centering
    \includegraphics[scale=0.5]{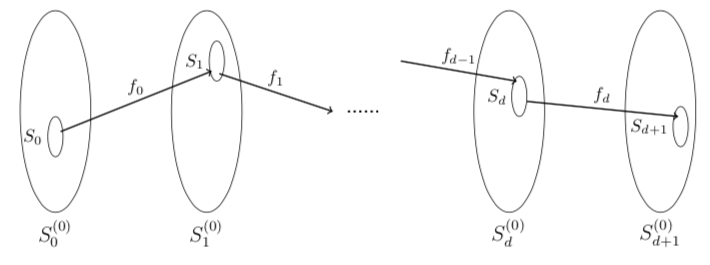}
    \caption{The shuffling oracle: $f_0,\dots,f_{d-1}$ are random one-to-one functions on a greater domain. For $x\in S_0$, $f(x) = f_d\circ\cdots\circ f_0(x)$.}
    \label{fig:good_oracle}
\end{figure}
It is worth noting that an algorithm can learn $S_i$ after the $i$-th query with probability~one, and thus is able to evaluate $f$ after $d+1$ sequential queries. However, 
using at most $d$ sequential queries is still not enough for the algorithm to evaluate $f$.

To sum up, the oracle we consider is   as follows. Let $f$ be an arbitrary Simon  function. Choose a sequence of random one-to-one functions $f_0,\dots,f_{d-1}$ defined on much larger domains $S_0^{(0)},\dots,S_d^{(0)}$, respectively, and  let $f_d$  be a two-to-one function such that $f_d\circ\cdots\circ f_0(x) = f(x)$ for $x\in S_0$.
We define the \emph{shuffling oracle} of $f$ with respect to $f_1,\dots,f_{d-1}$ to be an oracle
that returns a value of $f_0,\dots,f_{d-1}$ or $f_d$  when queried.
Thus an algorithm with access to this oracle can access $f_0,\dots,f_d$. 
We will simply call it the \emph{shuffling oracle} of $f$ in the following without mentioning the underlying $f_0,\dots,f_{d-1}$.

 The \textbf{\em $d$-Shuffling Simon's Problem} ($\SSP{d}$) is a decision problem defined as follows: 
\begin{definition}[$\SSP{d}$ (Informal)]
Let $f$ be a random one-to-one function or a random Simon function. Given oracle access to the shuffling oracle of $f$ (as in Fig.~\ref{fig:good_oracle}), the problem is to decide whether $f$ is a Simon function or not.
\end{definition}

We summarize our results in the following theorem.  
\begin{theorem}[Informal]\label{thm:informal_1}
For any $d$, $\SSP{d}\in (\class{BQNC_{2d+1}^{BPP}})^{\mathcal{O}}\cap (\class{BPP^{BQNC_{2d+1}}})^{\mathcal{O}}$, but $\SSP{d}\notin (\class{BQNC_{d}^{BPP}})^{\mathcal{O}}\cup (\class{BPP^{BQNC_{d}}}^{\mathcal{O}})$, where $\mathcal{O}$ is the shuffling oracle of $f$.   \footnote{Following the definition of $\class{BQNC_d^{BPP}}$ and $\class{BPP^{BQNC_d}}$, the gate set is the collection of one- and two-qubit gates. However, our oracle separation actually works for any gate set even with unbounded fan-out gates. The main point is that the depth of the queries is less than the depth of the shuffling oracle.} 
\end{theorem}

Theorem~\ref{thm:informal_1} states that relative to the shuffling oracles, $(2d+1)$-QC schemes (resp., $(2d+1)$-CQ schemes)  are strictly more powerful than $d$-QC schemes (resp., $d$-CQ schemes). Namely, doubling the quantum circuit depth indeed gives the hybrid model more power, and this cannot be traded for classical computation.        

Theorem~\ref{thm:informal_1} immediately implies the following corollary, which states that Jozsa's conjecture is false relative to an oracle.

\begin{corollary}
[Informal]\label{thm:informal_2}
For $d = \log^{\omega(1)} n$, $\SSP{d}\in \class{BQP}^{\mathcal{O}}$, but $\SSP{d}\notin (\class{BPP^{BQNC}})^{\mathcal{O}}\cup (\class{BQNC^{BPP}})^{\mathcal{O}}$, where $\mathcal{O}$ is the shuffling oracle of $f$.
\end{corollary}




%
\subsection{Discussion and Open Problems}\label{sec:discussion}
\paragraph{Jozsa's and Aaronson's conjectures} The hybrid model Jozsa considered in~\cite{Jozsa05} is the small-depth MBQC, which is characterized by the $d$-QC scheme in this work. On the other hand, Aaronson considered the complexity class $\class{BPP^{BQNC}}$, which is corresponding to the $d$-CQ scheme defined in this paper.

Since we cannot even prove that $\class{BPP}\neq \class{BQP}$, an unconditional separation seems  unachievable  by the state-of-the-art techniques. Instead, we can expect an oracle separation or a conditional lower bound based on certain assumptions in cryptography or complexity. In this work, we obtain an oracle separation between $\class{BQP}$ and $\class{BPP^{BQNC}}\cup \class{BQNC^{BPP}}$. Therefore, one central question is: can we achieve a conditional separation based on plausible assumptions, or equivalently, can we instantiate our oracle? A natural approach is to instantiate it based on cryptographic assumptions, such as virtual-black-box (VBB) obfuscations. 
However, while it is possible to make our oracle efficient, since our oracle uses random shuffling to hide the Simon's function, we may not assume secure VBB obfuscation of our oracle given the known negative evidence by Bitansky et al.~\cite{Bitansky14}. It is not clear whether our oracle can be instantiated based on any reasonable cryptographic assumptions. 

Our work also address another critical question: can we trade classical computations for quantum circuit depth? If this is possible, 
then it may be easier for near-term quantum computers to achieve quantum supremacy on practical problems. In this work, we give some negative evidence by showing a fine-grained depth separation result that $(\class{BPP^{BQNC_d}})^{\Ora}\neq(\class{BPP^{BQNC_{2d+1}}})^{\Ora}$. This result implies that doubling the quantum depth gives the model more power, and this cannot be traded by classical computations in the relativized world. We would like to see if we can obtain a sharper separation, such as $d$ versus $d+1$, under plausible assumptions in cryptography or complexity.   


\subsection{Acknowledgement}
 We are grateful to Scott Aaronson for bringing us to the problem of separating $\class{BQP}$ and $\class{BPP^{BQNC}}$. We also thank him for helpful discussions and valuable comments on
our manuscript. We would like to thank Richard Jozsa for clarifying his conjecture and Aaronson's conjecture. We acknowledge conversations with Matthew Coudron for explaining their related work to us. 

NHC's research is supported by Scott Aaronson's Vannevar Bush Faculty Fellowship from the US Department of Defense.
KMC's research is partially supported by the Academia Sinica Career Development Award under Grant  23-17 MOST QC project, and  MOST QC project under Grant MOST 107-2627-E-002-002.
CYL was financially supported from the Young Scholar Fellowship Program by Ministry of Science and Technology (MOST) in Taiwan, under
		Grant MOST108-2636-E-009-004.

\subsection{Independent Work}

Independent and concurrent to our work, Coudron and Menda also investigated Jozsa's conjecture and proved that the conjecture is false relative to some oracle~\cite{CM19}. However, they did not establish the sharp separation between quantum depth $d$ vs. $2d+1$ as we did. Interestingly, the oracle problem used in their work as well as the analysis are very different from ours, which may lead to incomparable extensions or applications in the future. In the following, we provide a high-level discussion and compare the two works.

At a high-level, Coudron and Menda use the Welded Tree Problem of Childs et al.~\cite{CC03} as the oracle problem to disprove Jozsa's conjecture, which leverage the separation between quantum walk and classical random walk. To show the separation, they do not need to modify the oracle problem, and the crux is a simulation argument showing that low depth quantum computation can be simulated by (inefficient) classical computation without blowing up the number of oracle queries too much. Hence, the hardness of the hybrid models follows by the classical hardness. The simulation argument relies on the structure of the welded tree oracle and does not seem to apply to our oracle.  

In contrast, our starting point is Simon's problem, which can be viewed as a separation between quantum depth $0$ (i.e., classical computation) and $1$ in the hybrid models, and our main idea is to lift the separation by pointer chasing and domain hiding. This allows us to prove the sharp separation between quantum depth $d$ vs. $2d+1$. To prove the separation, we generalize techniques in cryptography (from both quantum and classical crypto literature) to show that small quantum depth is not enough to find the hidden domain and hence the hybrid models with quantum depth $d$ are not capable to solve the $\SSP{d}$ problem.

We mention that the original version of Coudron and Menda~\cite{CM19} only explicitly considered the hybrid model $d$-CQ schemes, but their analysis extends directly to establish hardness of the Welded Tree Problem for $d$-QC schemes. Furthermore, we believe that the techniques in both works extend to establish hardness of respective oracle problems for the natural hierarchy of hybrid models like $\class{{BPP^{BQNC}}^{BPP^{\cdots}}}$~\cite{CM19_2}.

%% file: overview.tex
\section{Proof Overview}


The main idea behind our shuffling oracle is combining \emph{pointer chasing} and \emph{domain hiding} to lift the hardness of Simon's problem. We start with the intuition behind our construction.

\paragraph{Pointer Chasing \& Domain Hiding}The pointer chasing is to allow the function to be evaluated by using $d+1$ parallel queries, but make it hard by using at most $d$ parallel queries. To implement this idea, we decompose a Simon function $f$ as $f_0,\dots,f_d$, where $f_0,\dots,f_{d-1}$ are random ``shufflings'' of the original function and $f_d$ is the function such that $f = f_d\circ\cdots\circ f_0$. Ideally, to evaluate $f$, a standard approach is to evaluate the functions in sequence. However, it is unclear if there are clever methods. To tackle this, our second idea is hiding the original domains $S_0,\dots,S_d$ in larger domains $S_0^{(0)},\dots,S_d^{(0)}$, respectively. Intuitively, to evaluate $f$, one needs to learn $f_d$ in $S_d$, or say, needs to find $S_d$ in $S_d^{(0)}$. Note that it suffices to use $i$ rounds of parallel queries to learn $S_i$. The goal of the shuffling oracle is to obtain the (tight) hardness that it prevents the algorithm from learning $S_d$ before $d$th parallel queries so that one cannot evaluate $f_d$ in the $d$th parallel queries.


\paragraph{Obtaining Indistinguishability from Hiding} The intuition behind our analysis is similar to the One-way-to-Hiding (O2H) Lemma~\cite{AHU18,unruh15} in quantum cryptography, and our analysis can be viewed as generalizing the lemma to the setting of the shuffling oracle against both hybrid models.
Briefly, the O2H Lemma shows that given any two random functions which have same function values on most elements except for some \emph{hidden subset}, we can bound the probability of distinguishing these two functions by the probability that the queries are in the hidden set, which we call the \emph{finding probability}. This probability will be negligible if the subset is a small enough random subset. The O2H Lemma implies that given a random function, we can program the mappings in the hidden subset arbitrarily. In our context, a natural approach is to argue that $S_d$ is a hidden subset in $S_d^{(0)}$ and hence we can indistinguishably program the mappings in $S_d$ in a way such that the programmed oracle has no information about $f$.  However, different from the typical use of O2H Lemma in cryptography, in the shuffling oracle, the algorithm can find $S_d$ by using $d$ parallel queries starting from $S_0$ (and query $S_d$ at the $d+1$st query), so we need to be careful in arguing that $S_d$ is a hidden subset.


\begin{figure}
    \centering
    \includegraphics[scale=0.7]{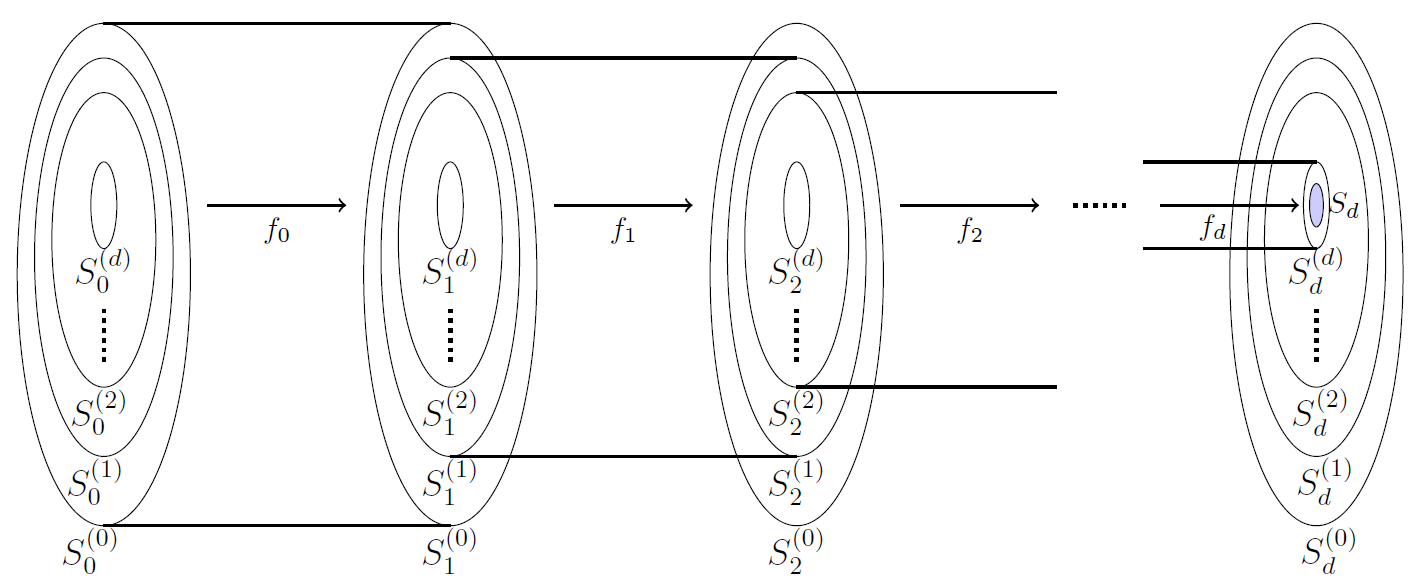}
    \caption{The domains partitions under the Russian-nesting-doll trick: $f_0$ on $S_0^{(0)}$ and $f_1,\dots,f_d$ on $S_1^{(0)}\setminus S_1^{(1)}, \dots,S_d^{(0)}\setminus S_d^{(1)}$ can be viewed as the ``first doll'', which will be opened by the first round parallel queries. Along this line, $f_{i-1}$ on $S_{i-1}^{(i-1)}$ and $f_{i},\dots,f_d$ on $S_i^{(i-1)}\setminus S_{i}^{(i)},\dots,S_d^{(i-1)}\setminus S_d^{(i)}$ can be viewed as the $i$th doll, which will be opened by the $i$th round parallel queries.}
    \label{fig:doll-2}
\end{figure}

\paragraph{Russian-Nesting-Doll Trick} 
We bypass the difficulty above by partitioning the large set $S_d^{(0)}$ into a series of smaller subsets $S_{d}^{(1)},\dots,S_{d}^{(d)}$, such that $S_d^{(\ell)}$ is a exponentially small subset of $S_d^{(\ell-1)}$ for $\ell\in [d]$, and for the ease of illustrating ideas, $S_d\subset S_d^{(d)}$\footnote{For the proof of $d$-QC and $d$-CQ scheme, there will be a small subset in $S_d$ which is not in $S_d^{(d)}$; however, we will show in the proof that it suffices to prove the oracle separation.}. We do such partition for all domains of $f_0,\dots,f_d$. Then, we use the O2H Lemma to show that after the first parallel query, the algorithm can learn $f_0$ on $S_0^{(0)}$ and $f_1,\dots,f_d$ on domains out of $S_1^{(1)},\dots,S_d^{(1)}$; however, it learns nothing about $f_1,\dots,f_d$ on $S_1^{(1)},\dots,S_d^{(1)}$. For the second parallel queries, we start from $f_1,\dots,f_d$ on $S_1^{(1)},\dots,S_d^{(1)}$, and use the O2H Lemma again to argue that the algorithm learns little about $f_2,\dots,f_d$ on $S_2^{(2)},\dots,S_d^{(2)}$. Finally, for the $d$th parallel queries, the algorithm has no information about $f_d$ on $S_d^{(d)}$ and thus has no information about $f$.  We can think of this approach as a Russian-nesting-doll trick: the mappings $f_0$ on $S_0^{(0)}$ and $f_1,\dots,f_d$ on domains out of $S_1^{(1)},\dots,S_d^{(1)}$ are the largest doll outside, there are $d$ dolls, and $f_d$ on $S_d^{(d)}$ is the smallest one we want to see. For an algorithm to see $S_d$, it needs to open all the dolls, and each requires one round parallel queries; therefore, the algorithm cannot see the smallest doll before $d$th parallel queries. We illustrate the idea of the Russian-nesting-doll trick as in Fig.~\ref{fig:doll-2}.


We give a brief overview of how to formalize the ideas above in the follows. 
\subsection{\texorpdfstring{$\class{QNC_d}$}{Lg} circuit}

This is the warm-up case. Let $U:=U_1,\dots,U_{d+1}$ be a $d$-depth quantum circuit which can make superposed queries to the shuffling oracle $\vecf:= (f_0,\dots,f_d)$ after $U_i$ for $i\in [d]$. We represent the computation as 
\begin{align*}
    U_{d+1} \vecf U_d \cdots \vecf U_1. 
\end{align*}
Following the Russian-nesting-doll trick, we partition the domains $S_0^{(0)},\dots,S_d^{(0)}$ into $d$ subsets. Then, for the $i$th parallel queries, we construct a \emph{shadow} oracle $\vecg_i$ of $\vecf$ such that $\vecg_i$ agrees with $\vecf$ on elements out of $S_j^{(j)},\dots,S_d^{(j)}$, but maps the rest to $\bot$ (a specific notation for no information). Here, we need to choose the subsets in a specific way, which we will describe in the formal proof. By using hybrid arguments and O2H Lemma, we inductively prove that $U_{d+1} \vecf U_d \cdots \vecf U_1$ is indistinguishable from $U_{d+1} \vecg_d U_d \cdots \vecg_1 U_1$.

\subsection{\texorpdfstring{$d$}{Lg}-QC scheme}

For the $d$-QC scheme, since the classical algorithm can make polynomial classical parallel queries, the set $S_d$ can be found by the classical algorithm. However, the classical algorithm can only make polynomially many queries, and the rest of the points which are not queried are still uniform.  Therefore, we can remove whatever points learned by classical queries from the hidden sets $S^{(j)}_k$, and use O2H Lemma and the hybrid arguments as before. At the end, the algorithm does not learn $S_d$ except for polynomially many points, which is unlikely to learn the shift  (as in the proof of classical hardness for Simons).

\subsection{\texorpdfstring{$d$}{Lg}-CQ scheme} 

Arguing that $S_d$ is a hidden subset in $S_d^{(0)}$ against the $d$-CQ scheme is the most challenging part of our analysis.
The main obstacle is that the measurement outcome of the quantum circuit can contain global information about the shuffling oracle, and this makes it hard to reason that the shuffling oracle is still uniformly random conditioned on the measurement outcome. However, since the measurement outcomes can only be a ``short classical advice'', intuitively, it cannot contain too much information about the shuffling oracle. 

To formalize this intuition, we (partially) generalize a presampling argument in (classical) cryptography~\cite{CDG18,unruh07}. 
Informally, the argument states that a random function conditioned on a short classical advice string is indistinguishable from a \textsl{convex combination} of random functions that are fixed on a few elements for any classical algorithm making polynomial queries. Ideally, we would like to generalize it to the shuffling oracle (which is more complicated than random oracle) and show indistinguishability to a convex combination of shuffling oracles with a few points fixed for quantum algorithms with a bounded number of queries. Unfortunately, we do not know how to achieve this, but we prove something weaker that is sufficient: we show that given a short classical advice string, the shuffling oracle is indistinguishable from a convex combination of ``\emph{almost-uniform}'' shuffling oracles with few points fixed for any quantum algorithm making \emph{one quantum parallel query}.

This weaker statement suffices to prove our oracle separation.
While there are some further subtleties, the main idea is that we apply it inductively after each quantum parallel query, and the short advice string only fixes polynomially many random points in the shuffling oracle. Given this, we are able to show that $S_d$ remains a hidden set against the $d$-CQ scheme and establish hardness of $\SSP{d}$ against $\class{BPP^{BQNC_d}}$.

%% file: definition.tex
\section{Preliminaries}\label{sec:preliminaries}

In this section, we first introduce the distance measures of quantum states. Then, we give formal definitions of the $d$-CQ and $d$-QC schemes.

\subsection{State distance}

\begin{definition}
Let $\mathcal{H}$ be a Hilbert space. For any two pure states $\ket{\psi},\ket{\phi}\in \mathcal{H}$, we define
\begin{itemize}
    \item (Fidelity) $F(\ket{\psi}, \ket{\phi}) := |\ipro{\psi}{\phi}|$;
    \item (Two-norm distance) $\|\ket{\psi} - \ket{\phi}\|$.
\end{itemize}
\end{definition}

Then, we define distance measures between mixed states.
\begin{definition}
Let $\mathcal{H}$ be a Hilbert space. For any two mixed states $\rho,\rho'\in \mathcal{H}$, 
\begin{itemize}
    \item (Fidelity) $F(\rho,\rho') := tr(\sqrt{\sqrt{\rho} \rho'\sqrt{\rho}})$. 
    \item (Trace distance) $TD(\rho,\rho'):=\frac{1}{2} \Tr|\rho - \rho'|$. 
    \item (Bures distance) $B(\rho,\rho') := \sqrt{2-2F(\rho,\rho')}$.
\end{itemize}
\end{definition}

The probability for a quantum procedure to distinguish two states can be bounded by the Bures distance between the two states.
\begin{claim}\label{lem:prob_to_norm}
For any two mixed states $\rho$ and $\rho'$, for any quantum algorithm $\A$ and for any classical string $s$,
\[
    |\Pr[\A(\rho)=s] - \Pr[\A(\rho')=s]| \leq B(\rho,\rho').
\]
\end{claim}
\begin{proof}
It is well-known that $|\Pr[\A(\rho)=s] - \Pr[\A(\rho')=s]|\leq \frac{1}{2}\Tr|\rho-\rho'|$. Then, 
\begin{eqnarray*}
TD(\rho,\rho') &\leq& \sqrt{1-F(\rho,\rho')^2} \\
&=& \sqrt{\frac{1+F(\rho,\rho')}{2}}\sqrt{2-2F(\rho,\rho')} \\
&\leq& B(\rho,\rho'). 
\end{eqnarray*}
\end{proof}

The state distance and Claim~\ref{lem:prob_to_norm} will be used shortly in the following sections.  

\subsection{Computational Model}
We here define two schemes which interleave low-depth quantum circuits and classical computers. The first scheme is called $d$-depth quantum-classical scheme ($d$-QC scheme) and the second scheme is $d$-depth classical-quantum scheme ($d$-CQ scheme).

We say a set of gates forms a layer if all the gates in the set operate on disjoint qubits. Gates in the same layer can be parallelly applied.  We define the number of layers in a circuit as the {\emph depth} of the circuit. In the following, we define circuit families which has circuit depth $d$ as in~\cite{TD04,MN98}.
\begin{definition}[$d$-depth quantum circuit $\class{QNC_d}$]\label{def:d_qc}
A $\class{QNC_d}$ quantum circuit family $\{C_n: n>0\}$ is defined as below: 
\begin{itemize}
    \item There exists a polynomial $p$ such that for all $n>0$, $C_n$ operates on $n$ input qubits and $p(n)$ ancilla qubits; 
    \item for $n>0$, $C_n$ has the initial state $\ket{0^{n+p(n)}}$, consists of $d$ layers of one- and two-qubit gates, and measures all qubits after the last layer.
\end{itemize}
\end{definition}

We can illustrate a $\class{QNC_d}$ circuit as in Fig.~\ref{fig:qnc}, where $U_i$ for $i\in [d]$ is a unitary which can be implemented by one layer of one- and two-qubit gates, and the last computational unit is a qubit-wise measurement in the standard basis. 

\begin{figure}
    \centering
    \includegraphics[scale=0.5]{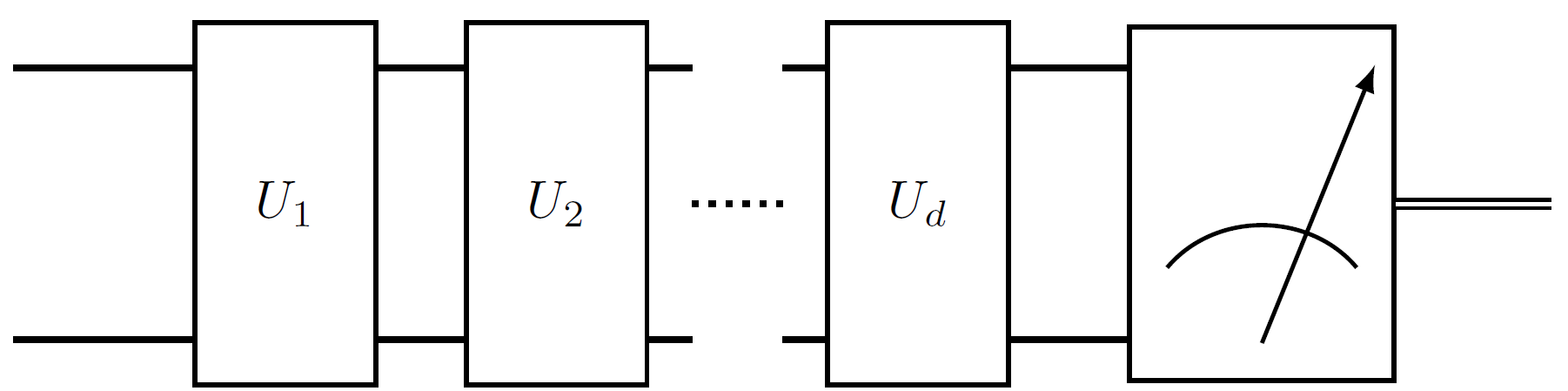}
    \caption{The $\class{QNC_d}$ circuit: The single-line wires are the quantum wires, and the double-line wire is the classical wire.}
    \label{fig:qnc}
\end{figure}

In some studies, $\class{QNC}$ is also used to refer the set of languages decided by the quantum circuits. For clarity, we define the set of languages decided by $\class{QNC_d}$ as $\class{BQNC_d}$ as follows: 
\begin{definition}[$\class{BQNC_d}$]
The set of languages $L=\{L_n: n>0\}$ for which there exists a circuit family $\{C_n: n>0\}\in QNC_d$ such that for $n>0$, for any $x$ where $|x| =n$, 
\begin{itemize}
    \item if $x\in L_n$, then $\Pr[C_n(x)=1]\geq 2/3$; 
    \item otherwise, $\Pr[C_n(x)=1]\leq 1/3$. 
\end{itemize}
\end{definition}

Then we define the quantum analogue of Nick's class
\begin{definition}[$\class{BQNC^k}$]
The set of languages $L=\{L_n: n>0\}$ for which there exists a circuit family $\{C_n: n>0\}\in QNC_d$ for $d=O(\log^k n)$ such that for $n>0$, for any $x$ where $|x| =n$, 
\begin{itemize}
    \item if $x\in L_n$, then $\Pr[C_n(x)=1]\geq 2/3$; 
    \item otherwise, $\Pr[C_n(x)=1]\leq 1/3$. 
\end{itemize}
\end{definition}

For simplicity, we define $\class{BQNC}$ as the set of languages which can be decided by $\poly\log$-depth quantum circuit that is $\class{BQNC} := \bigcup_{k}\class{BQNC}^{k}$.


A $d$-depth quantum circuit with oracle access to some classical function $f$ is a sequence of operations
\begin{eqnarray}
 U_1 \xrightarrow{q} f \xrightarrow{q} U_2\xrightarrow{q} f\xrightarrow{q} \cdots U_d \xrightarrow{q} f \xrightarrow{q} U_{d+1}, \label{eq:prem_1}
\end{eqnarray}
where $U_i$'s are unitaries as we have defined in~\ref{def:d_qc}. The transition $\xrightarrow{q}$ implies that $U_i$ can send and receive quantum messages from $f$.  We add an additional layer of unitary $U_{d+1}$ to the computational model to process the information from the last call of $f$.

In the quantum query model, we usually consider $f$ as an unitary operator $U_{f}$
\begin{align*}
    U_{f}\ket{x}\ket{0} = \ket{x}\ket{f(x)}. 
\end{align*}
We represent the process in Eq.~\ref{eq:prem_1} as follows: let $\ket{0,0}_{QA}\ket{0}_W$ be the initial state, where registers $Q$ and $A$ consist of qubits the algorithm sends to or receives from the oracle and register $W$ consists of the rest of the qubits the algorithm holds as working space. Let 
\begin{align*}
    & U_1\ket{0,0}_{QA}\ket{0}_W = \sum_{x,y}c(x,y)\ket{x,y}_{QA}\ket{z(x,y)}_W
\end{align*}
be the quantum message before applying the first $U_f$. After applying $U_{f}$,  
\begin{align*}
    & U_f\sum_{x,y}c(x,y)\ket{x,y}_{QA}\ket{z(x,y)}_W = \sum_{x,y}c(x,y)\ket{x,y\oplus f(x)}_{QA}\ket{z(x,y)}_W. 
\end{align*}
We represent Eq.~\ref{eq:prem_1} as a sequence of unitaries $U_{d+1}U_fU_d  \cdots  U_fU_1$ in this way. For simplicity, we rewrite Eq.~\ref{eq:prem_1} as $U_{d+1}fU_d  \cdots  fU_1$ in the rest of the paper.


\subsubsection{The \texorpdfstring{$d$}{Lg}-QC Scheme and \texorpdfstring{$\class{BQNC_d^{BPP}}$}{Lg}}
The first scheme we consider is the $d$-QC scheme, which is a generalized model for d-depth MBQC and can be represented as the following sequence 
\begin{eqnarray}
\A_c \xrightarrow{c} (\Pi_{0/1}\otimes I)\circ U^{(1)} \xrightarrow{c} \A_c \xrightarrow{c/q}  \cdots \xrightarrow{c/q} (\Pi_{0/1}\otimes I)\circ U^{(d)} \xrightarrow{c} \A_c, \label{eq:d-qc-model_1}
\end{eqnarray}
where $\A_c$ is a randomized algorithm, $U^{(i)}$ is a single-depth quantum circuit, and $\Pi_{0/1}$ is a projective measurement in the standard basis on a subset of the qubits. The arrows $\xrightarrow{c}$ and $\xrightarrow{q}$ indicate the classical and quantum messages transmitted. We illustrate the scheme as in Fig.~\ref{fig:qnc_bpp}
\begin{figure}
    \centering
    \includegraphics[width=\textwidth]{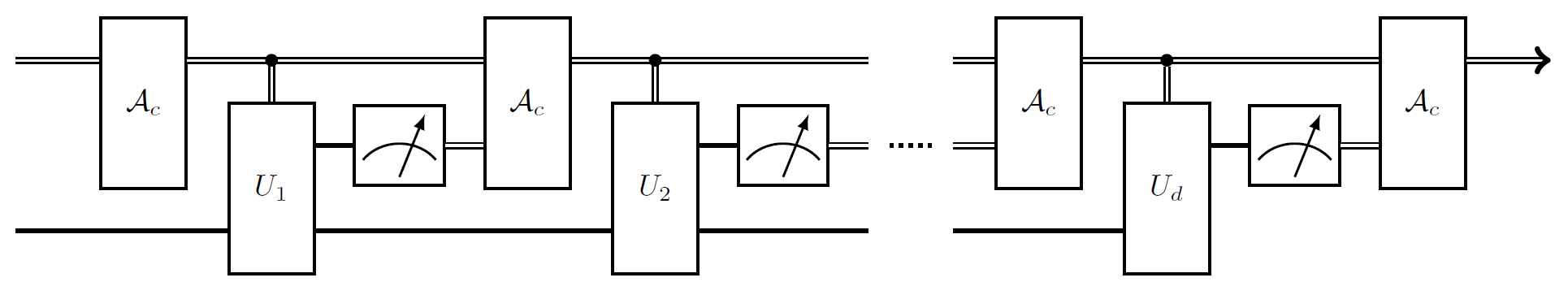}
    \caption{The $d$-depth quantum-classical ($d$-QC) scheme: The single-line wires stands for the quantum wires, and the double-line wires are for classical wires.}
    \label{fig:qnc_bpp}
\end{figure}

In this model, after applying $U^{(i)}$, one measures a subset of qubits, uses the measurement outcome to decide $U^{(i+1)}$ by the classical polynomial algorithm $\A_c$, and then apply $U^{(i+1)}$ to the rest of the qubits. Let $L^{(i)}$ be the procedure $\A_c\rightarrow{c} (\Pi_{0/1}\otimes I)\circ U^{(i)} $. Then, we rewrite Eq.~\ref{eq:d-qc-model_1} as 
\begin{eqnarray}
L^{(1)} \xrightarrow{c/q} \cdots \xrightarrow{c/q} L^{(d)}\xrightarrow{c}\A_c . \label{eq:d-qc-model_2}
\end{eqnarray}

Then, we define the languages which can be decided by $d$-QC scheme. 
\begin{definition}[$\class{BQNC_d^{BPP}}$]
The set of languages $\Lang=\{\Lang_n: n>0\}$ for which there exists a family of $d$-QC schemes $\{\A_n: n>0\}$ such that for $n>0$, for any $x$ where $|x| =n$, 
\begin{itemize}
    \item if $x\in \Lang_n$, then $\Pr[\A_n(x)=1]\geq 2/3$;
    \item otherwise, $\Pr[\A_n(x)=1]\leq 1/3$. 
\end{itemize} 
\end{definition}

Let $\A$ be a $d$-QC scheme with access to some oracle $\Ora$. We represent $\A^{\Ora}$ as a sequence of operators: 
\begin{eqnarray}
  (L^{(1)})^{\Ora} \xrightarrow{c/q} \cdots \xrightarrow{c/q} (L^{(d)})^{\Ora}\xrightarrow{c}\A_c^{\Ora},  
\end{eqnarray}
where $(L^{(i)})^{\Ora}:= \A_c^\Ora \xrightarrow{c}(\Pi_{0/1}\otimes I)\circ \Ora U^{(i)}$. We illustrate the scheme as in Fig.~\ref{fig:qnc_bpp_o}.
\begin{figure}
    \centering
    \includegraphics[width=\textwidth]{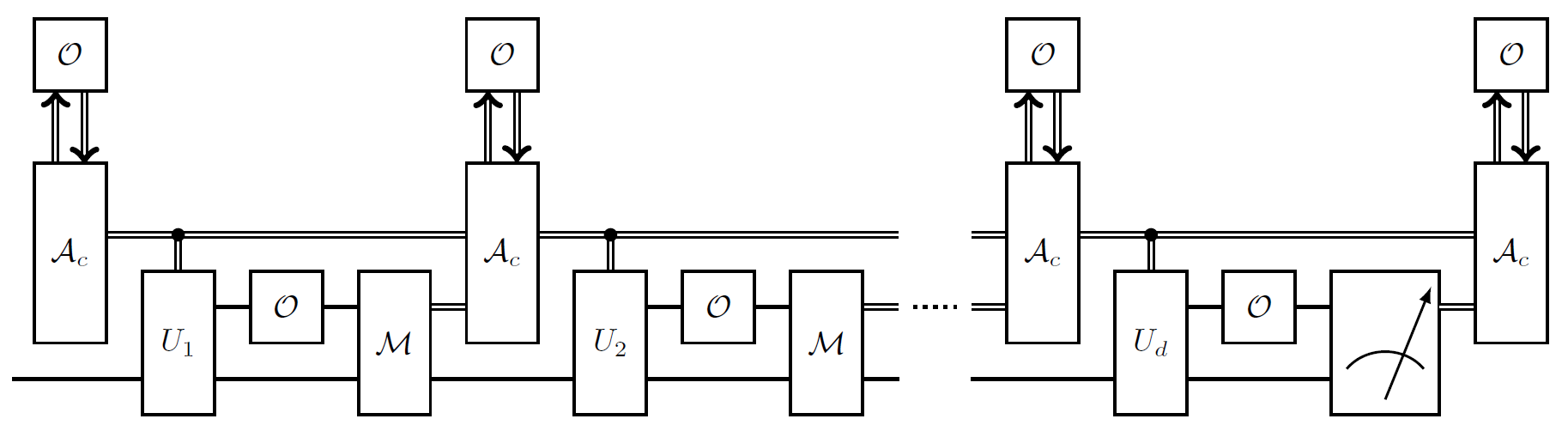}
    \caption{The $d$-QC scheme with access to an oracle $\Ora$}
    \label{fig:qnc_bpp_o}
\end{figure}

\begin{definition}[$\class{(BQNC_d^{BPP})}^{\Ora}$]
The set of languages $\Lang^{\Ora}:=\{\Lang^{\Ora}_n: n>0\}$ for which there exists a family of $d$-QC schemes $\{\A^{\Ora}_n: n>0\}$ such that for $n>0$, for any $x$ where $|x| =n$, 
\begin{itemize}
    \item if $x\in \Lang^{\Ora}_n$, then $\Pr[\A^{\Ora}_n(x)=1]\geq 2/3$; 
    \item otherwise, $\Pr[\A^{\Ora}_n(x)=1]\leq 1/3$. 
\end{itemize} 
\end{definition}

Similar to the definition of $\class{BQNC}$, we define $\class{BQNC^{BPP}}$ as a set of languages which can be decided by a family of $d$-QC schemes with $d=O(\poly\log n)$. 

\subsubsection{The \texorpdfstring{$d$}{Lg}-CQ scheme and \texorpdfstring{$\class{BPP^{BQNC_d}}$}{Lg}}
Here we define the $d$-CQ scheme that is a classical algorithm which has access to a $\class{QNC}_d$ circuit during the computation.  We represent the scheme as follows: 
\begin{eqnarray}
\A_{c,1} \xrightarrow{c} \Pi_{0/1}\circ C \xrightarrow{c} \cdots \xrightarrow{c}  \A_{c,m} \xrightarrow{c} \Pi_{0/1}\circ C \xrightarrow{c} \A_{c,m+1},\label{eq:d-cq-model_1}
\end{eqnarray}
where $m=\poly(n)$, $A_{c,i}$ is a randomized algorithm, $C$ is a $d$-depth quantum circuit, and $\Pi_{0/1}$ is the standard-basis measurement. We illustrate the scheme as in Fig.~\ref{fig:bpp_qnc}. 
\begin{figure}
    \centering
    \includegraphics[width=\textwidth]{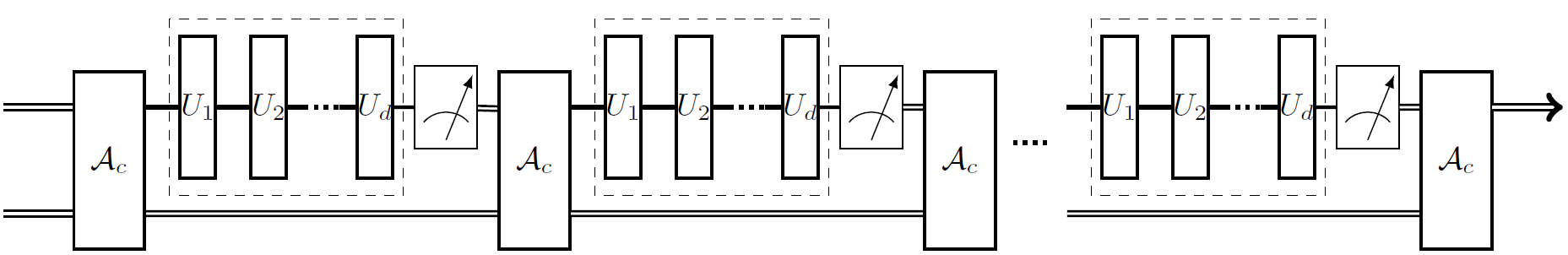}
    \caption{The $d$-depth classical-quantum ($d$-CQ) scheme}
    \label{fig:bpp_qnc}
\end{figure}

In this scheme, the classical algorithm can perform queries to the $\class{QNC}_d$ circuit, and then use the measurement outcomes from the circuit as part of the input to the following procedures. We let $L^{(i)}:= \A_{c,i} \xrightarrow{c} \Pi_{0/1}\circ C $ and rewrite Eq.~\ref{eq:d-cq-model_1} as 
\[
    L^{(1)}\xrightarrow{c} L^{(2)}\xrightarrow{c}\cdots\xrightarrow{c} L^{(m)} \xrightarrow{c} \A_{c,m+1}. 
\]

\begin{definition}[$\class{BPP^{BQNC_d}}$]
The set of languages $\Lang=\{\Lang_n: n>0\}$ for which there exists a family of $d$-CQ schemes $\{\A_n: n>0\}$ such that for $n>0$, for any $x$ where $|x| =n$, 
\begin{itemize}
    \item if $x\in \Lang_n$, then $\Pr[\A_n(x)=1]\geq 2/3$;
    \item otherwise, $\Pr[\A_n(x)=1]\leq 1/3$. 
\end{itemize} 
\end{definition}

Let $\A$ be a $d$-CQ scheme with access to some oracle $\Ora$. We represent $\A^{\Ora}$ as 
\begin{eqnarray}
(L^{(1)})^{\Ora}\xrightarrow{c} (L^{(2)})^{\Ora}\xrightarrow{c}\cdots\xrightarrow{c} (L^{(m-1)})^{\Ora} \xrightarrow{c} (\A_c^{(m)})^{\Ora}\label{eq:d-cq-model_2}
\end{eqnarray}
where $(L^{(i)})^{\Ora}:= (\A_{c,i})^{\Ora} \xrightarrow{c} \Pi_{0/1}\circ (U^{(d)}\Ora U^{(d-1)}\cdots \Ora U^{(1)})$. We illustrate the scheme as in Fig.~\ref{fig:bpp_qnc_o}. 
\begin{figure}
    \centering
    \includegraphics[width=\textwidth]{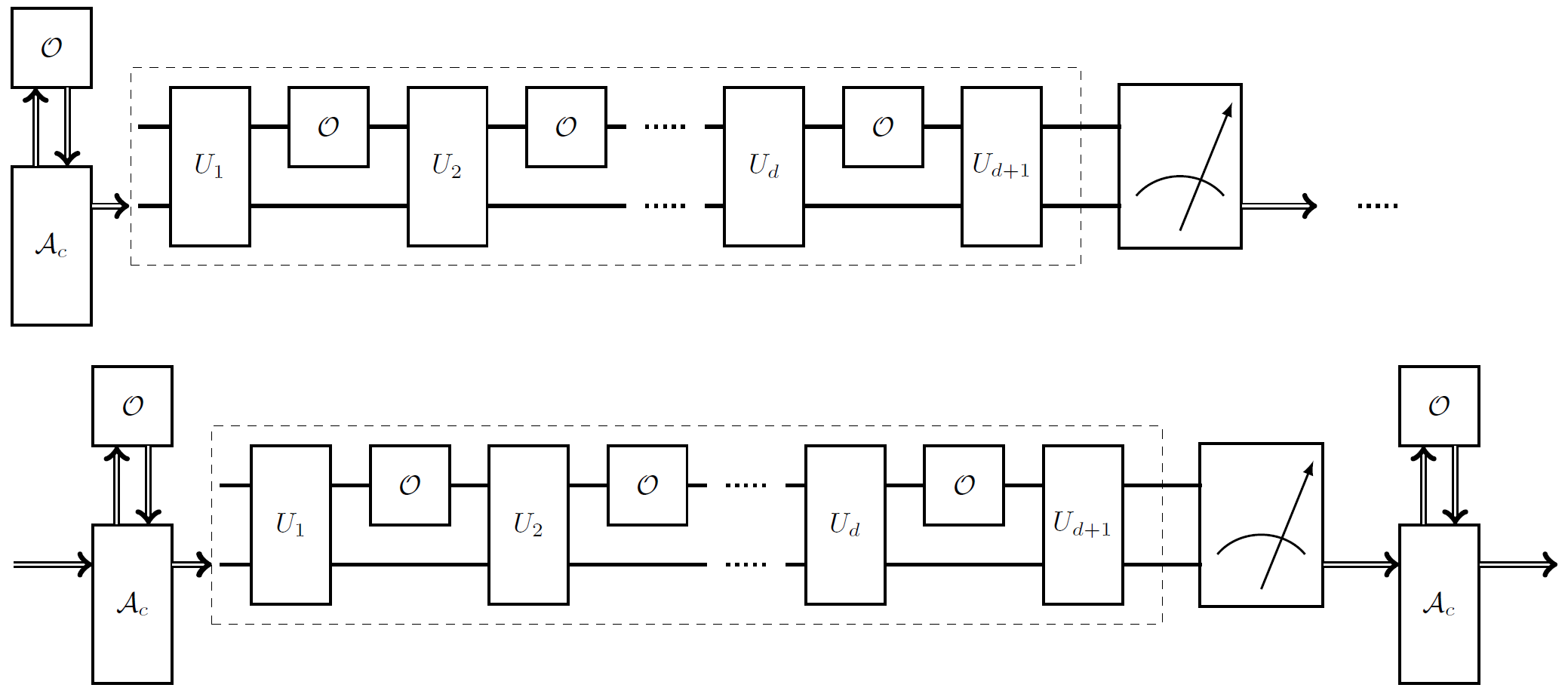}
    \caption{The $d$-CQ scheme with access to an oracle $\Ora$}
    \label{fig:bpp_qnc_o}
\end{figure}

\begin{definition}[$(\class{BPP^{BQNC_d}})^{\Ora}$]
The set of languages $\Lang^{\Ora}=\{\Lang^{\Ora}_n: n>0\}$ for which there exists a family of $d$-CQ schemes $\{\A^{\Ora}_n: n>0\}$ such that for $n>0$, for any $x$ where $|x| =n$, 
\begin{itemize}
    \item if $x\in \Lang^{\Ora}_n$, then $\Pr[\A^{\Ora}_n(x)=1]\geq 2/3$; 
    \item otherwise, $\Pr[\A^{\Ora}_n(x)=1]\leq 1/3$. 
\end{itemize} 
\end{definition}
We define $\class{BPP^{BQNC}}$ as a set of languages which can be decided by a family of $d$-CQ schemes with $d=O(\polylog n)$. 

The main differences between $d$-CQ and $d$-QC schemes are that 1) the $d$-QC scheme can transmit quantum messages from one layer to the next, but the $d$-CQ scheme can only send classical messages, and 2) a $d$-QC scheme has at most $d$ layers, but a $d$-CQ scheme may have $m\times d$ layers. According to these observations, these two schemes seem to be incomparable.  

The notations $\class{BQNC^{BPP}}$ and $\class{BPP^{BQNC}}$ may not be standard in quantum complexity theory. For instance, one may expect the quantum circuit can make superposed queries, but we only allow classical queries. The reason we chose these non-standard names is that $\class{BPP^{BQNC}}$ and $\class{BQNC^{BPP}}$ are more intuitive for readers to capture the ideas of the two models and can help readers to figure out the differences between these two models as we have mentioned in the previous paragraph.

%% file: Sora.tex
\section{The \texorpdfstring{$d$}{Lg}-shuffling Simon's Problem (\texorpdfstring{$\SSP{d}$}{Lg})}\label{sec:shuffling_oracle_def}
 
In this section, we define the oracle and the corresponding oracle problem which separates the general quantum algorithms from $d$-QC and $d$-CQ schemes. We call the oracle the \emph{shuffling oracle} and the corresponding problem as the $d$-shuffling Simon's problem ($\SSP{d}$).   

\subsection{The Shuffling Oracle}
For $X$ and $Y$ any two sets, we let $P(X,Y)$ be the set of one-to-one functions from $X$ to $Y$. For example, $\mathcal{P}(\Z_2^n,\Z_2^n)$ is the set of all permutations over $\Z_2^n$.  Let $f: \Z_2^n\rightarrow  \Z_2^n$ be an arbitrary function and $d\in \mathbb{N}$.  We define the $d$-depth shuffling of $f$ as below.  
\begin{definition}[$(d,f)$-Shuffling]\label{def:shuffle}
A $(d,f)$-shuffling of $f$ is defined by $\vecf:=(f_0,\dots,f_d)$, where $f_0,\dots,f_{d-1}\in P(\Z_2^{(d+2)n},\Z_2^{(d+2)n})$ are chosen randomly, and then we choose $f_d$ be the function satisfying the following properties: Let  $S_d:=\{f_{d-1}\circ\cdots\circ f_0(x'):\; x'=0,\dots,2^n-1\}$.  
\begin{itemize}
    \item For $x\in S_d$, let $f_{d-1}\circ\cdots\circ f_0(x')=x$, and we choose the function $f_d:S_d\rightarrow [0,2^n-1]$ satisfying that $f_d\circ f_{d-1}\circ\cdots\circ f_0(x') = f(x')$.
    \item For $x\notin S_d$, we let $f_d(x)=\bot $. 
\end{itemize}
We let $\Sf(d,f)$ be the set of all $(d,f)$-shuffling functions of $f$. For simplicity, we denote $f_d$ on the subdomain $S_d$ as $f^*_d$. 
\end{definition}

In this paper, we consider the case that the $(d,f)$-shuffling is given randomly. One of the most natural way is sampling a shuffling function $\vecf$ uniformly at random from $\Sf(d,f)$. We describe the sampling procedure as below.
\begin{definition}[$\D(f,d)$]\label{def:uniform_distribution}
Draw $f_0,\dots,f_{d-1}$ uniformly at random from $\mathcal{P}(\Z_2^{(d+2)n},\Z_2^{(d+2)n})$ and then choose $f_d^*$ such that $f^*_d\circ\cdots\circ f_0(x) = f(x)$ for $x\in \Z_2^n$.
\end{definition}

Fix $d\in \mathbb{N}$ and a function $f:\Z_2^n\rightarrow \Z_2^n$, we can define a random oracle which is a $(d,f)$-shuffling chosen uniformly randomly from $\Sf(d,f)$. 
\begin{definition}[Shuffling oracle $\Oru$]
Let $f$ be an arbitrary function from $\Z_2^n$ to $\Z_2^n$. Let $d\in \mathbb{N}$. We define the shuffling oracle $\Oru$ as a $(d,f)$-shuffling $\vecf$ chosen from $\Sf(d,f)$ according to $\D(f,d)$. 
\end{definition}

If we sample $\vecf$ according to the distribution in Def.~\ref{def:uniform_distribution}, only $f_d^*$ encodes the information of $f$, while $(f_0,\dots,f_{d-1})$ are just random one-to-one functions. We call $f_d^*$ the \textbf{core function} of $f$. 

For simplifying our proofs, we here define paths in the shuffling oracle. The concept of paths will be used when we consider $d$-QC and $d$-CQ schemes. 
\begin{definition}[Path in $\Oru$]\label{def:path}
Let $\vecf\in \Sf(d,f)$. we say $(x_0,\dots,x_{d+1})$ is a \textbf{path} in $\vecf$ if $f_i(x_i) = x_{i+1}$ for $i=0,\dots,d$.
\end{definition}

Now, we describe the oracle access to the shuffling oracle $\Oru$. Let $\ket{\phi}$ be the input state to $\Oru$, which we represent in the form  
\begin{eqnarray}
\ket{\phi}:=\sum_{\X_0,\dots,\X_d} c(\X_0,\dots,\X_d)\left(\bigotimes_{i=0}^d\ket{i,\X_i}\right)_{\regR_Q}\otimes \ket{0}_{\regR_N} \otimes \ket{w(\X_0,\dots,\X_d)}_{\regR_W},\label{eq:state_1}
\end{eqnarray}
where $\ket{w(\X_0,\dots,\X_d)}$'s are some arbitrary states and $\X_i$ is a set of elements in the domain of $f_i$. The queries in the register $\regR_Q$ and the ancillary qubits in the register $\regR_N$ will be processed by the oracle, while the remaining local working qubits in the register $\regR_W$ are unchanged and hold by the algorithm; the state $\ket{i,\X_i}$ denotes the set of parallel queries to function $f_i$. 

We let $\vecf\in \Sf(d,f)$ be sampled according to $\D(f,d)$, then
\[
    \vecf\ket{\phi}:= \left(\bigotimes_{i=0}^d\ket{i,\X_i}\ket{f_i(\X_i)}\right)_{\regR_Q,\regR_N} \otimes \ket{w(\X_0,\dots,\X_d)}_{\regR_W}.
\]
Applying $\Oru$ on $\ket{\phi}$ gives the mixed state
\[
    \Oru (\opro{\phi}{\phi}) := \sum_{\vecf\in \Sf(d,f)} \frac{1}{|\Sf(d,f)|} \vecf \opro{\phi}{\phi} \vecf.
\]

\subsection{Shuffling Simon's Problem}

We recall the definitions of Simon's function and Simon's problem. 
\begin{definition}[Simon's function] \label{def:Simon_fun}
A two-to-one function $f:\Z_{2}^n\rightarrow \Z_2^n$ for $n\in \mathbb{N}$ is a Simon's function if there exists an $s\in \Z_2^n$ such that $f(x) = f(x+s)$ for $x\in \Z_2^n$. 
\end{definition}

\begin{definition}[Simon's problem ]\label{def:Simon}
Let $n\in \mathbb{N}$. Let $\F$ be the set of all Simon's functions from $\Z_2^n$ to $\Z_2^{n}$. Choose $f$ from $\F$ uniformly at random, the problem is to find the hidden shift $s$ of $f$. 
\end{definition}

The decision version of Simon's problem is as follows: 
\begin{definition}[Decision Simon's problem]\label{def:Simon_2}
Let $n\in \mathbb{N}$. Let $\F$ be the set of all Simon's functions from $\Z_2^n$ to $\Z_2^{n}$. Choose $f$ to be  either a random Simon's function from $\F$ or a random one-to-one function from $\Z_2^n$ to $\Z_2^n$ with equal probability, the problem is to decide which case $f$ is. 
\end{definition}

Both problems have been shown to be hard classically and can be solved in quantum polynomial time. We define the $\SSP{d}$ by combining Simon's problem and the shuffling oracle. 

\begin{definition}[$d$-Shuffling Simon's Problem ($\SSP{d}$)]
Let $d\in \mathbb{N}$. Let $f:\{0,1\}^n\rightarrow \{0,1\}^{n}$ be a random Simon's function or a random one-to-one function with equal probability. Given access to the $(d,f)$-shuffling oracle $\Oru$ of $f$, the problem is to decide which case $f$ is.
\end{definition}\label{def:dssp}

The search version of the $\SSP{d}$ is as follows: 
\begin{definition}[Search $\SSP{d}$]
Let $d\in \mathbb{N}$. Let $f:\{0,1\}^n\rightarrow \{0,1\}^{n}$ be a random Simon's function. Given access to the $(d,f)$-shuffling oracle $\Oru$ of $f$, the problem is to find the hidden shift of $f$.
\end{definition}\label{def:dssp_2}

Then, we define an oracle $\Ora$ and a language $\Lang(\Ora)$ corresponding to the $\SSP{d}$.  
\begin{definition}~\label{def:language}
Let $\{f_i:\Z_2^i\rightarrow \Z_2^i;\; i>0\}$ be the set of functions satisfying the promise of the decision Simon's problem. Let $\Ora:=\{\Ora^{f_i,d(i)}_{unif}: i>0\}$, where $\Ora^{f_i,d(i)}_{unif}$ is the ($d(i),f_i$)-shuffling oracle for the function $f_i$ and $d(i)= i$. 
The language is defined as follows: 
\[
    \Lang(\Ora) := \{1^n: f_n\mbox{ is a Simon's function.}\} 
\]
\end{definition}

\begin{remark}
If the $\SSP{d}$ is intractable for any $d$-CQ or $d$-QC scheme for any $d$, then 
\[
\Lang(\Ora)\notin \class{(BPP^{BQNC})}^{\Ora}\cup \class{(BQNC^{BPP})}^{\Ora}.
\]
This is because $d(n)=\Theta(n)$ is asymptotically greater than $\log^k n$ for any constant $k$.  We can prove it by contradiction. If there is a $d$-CQ or $d$-QC algorithm can decide $\Lang(\Ora)$, then we can use that algorithm to solve the $\SSP{d}$. 

\end{remark}

We can show that the $\SSP{d}$ can be solved with a $2d+1$-depth quantum circuit. The idea is using Simon's algorithm and erasing the queries on the path.  
\begin{theorem}\label{thm:dSSP_solve}
The $\SSP{d}$ and the search $\SSP{d}$ can be solved by a $\class{QNC}_{2d+1}$ circuit with classical post-processing. 
\end{theorem}
\begin{proof}
The $\SSP{d}$ can be solved via Simon's algorithm. We show the proof here. 
\begin{eqnarray*}
    \sum_{x\in \Z_2^n}\ket{x}\ket{0,\dots,0} &\xrightarrow{f_0\sim\D_f}&  \sum_{x\in \Z_2^n}\ket{x}\ket{f_0(x),\dots,0}\nonumber\\
    &\xrightarrow{f_1\sim\D_f}& \sum_{x\in \Z_2^n}\ket{x}\ket{f_0(x),f_1(f_0(x)),\dots,0}\nonumber\\
    &\xrightarrow{f_d\sim\D_f}& \sum_{x\in \Z_2^n}\ket{x}\ket{f_0(x),f_1(f_0(x)),\dots,f(x)}\nonumber\\
    &\xrightarrow{Measure}& \frac{1}{\sqrt{2}}(\ket{x}\ket{f_0(x),\dots,f_{d-1}(\cdots f_{0}(x))}\nonumber\\
    &&+\ket{x+s}\ket{f_0(x+s),\dots,f_{d-1}(\cdots f_{0}(x+s))}) \ket{f(x)}\nonumber\\
    &\xrightarrow{\mbox{uncompute } f_0,\dots,f_{d-1}}& \frac{1}{\sqrt{2}}\left(\ket{x}+\ket{x+s}\right)\ket{f(x)}\\
    &\xrightarrow{QFT}& \frac{1}{\sqrt{2^n}}\sum_{j\in \Z_2^n} ((-1)^{x\cdot j} + (-1)^{(x+s)j}) \ket{j}.\label{eq:fourier_sampling}
\end{eqnarray*}
When $s\cdot j =0$, measuring the first register outputs $j$ with non-zero probability. Other other hand, when $s\cdot j = 1$, the probability that measurement outputs $j$ is zero. Therefore, by sampling $O(n)$ copies of $j$'s, one can find $s$ which is orthogonal to all $j$'s with high probability. 
\end{proof}

Theorem~\ref{thm:dSSP_solve} directly implies that the language defined in Def.~\ref{def:language} is in $\class{BQP}$ relative to $\Ora$.  

%% file: O2H_Sora.tex
\section{Analyzing the Shuffling Oracle}\label{sec:shuffling_oracle_analysis}

In this section, we are going to prove some properties related to the shuffling oracle. We first define as sequence of subsets which are in the domains of $f_0,\dots,f_d$. 
\begin{definition}[$\vecS$]
Let $S_0=\{0,\dots,2^n-1\}$. For $j=1,\dots, d$, let $S_{j+1}= f_{j}\circ f_{j-1}\circ \cdots \circ f_0(S_0)$. We define $\vecS := (S_0,\dots,S_{d})$. 
\end{definition}

We define a sequence of hidden sets $\Sb$ corresponding to the shuffling oracle.  
\begin{definition}[The hidden sets $\Sb$]\label{def:S}
Let $d,n\in \mathbb{N}$, $f: \Z_2^n\rightarrow \Z_2^{n}$, and $\vecf$ be a $(d,f)$-shuffling of $f$. Then, we define the sequence of hidden sets $\Sb=(\vecS^{(0)},\dots,\vecS^{(d)})$ as follows. 
\begin{enumerate}
    \item Let $S_j^{(0)}=\Z_2^{(d+2)n}$ for $j=0,\dots,d$. We define $\vecS^{(0)} = (S_0^{(0)},\dots,S_{d}^{(0)})$.
    \item For $\ell=1,\dots,d$, for $j=\ell,\dots, d$, we choose $S_j^{(\ell)}\subseteq S_j^{(\ell-1)}$ randomly satisfying that $\frac{|S_j^{(\ell)}|}{|S_j^{(\ell-1)}|}\leq \frac{1}{2^n}$, $f_j(S_{j-1}^{(\ell)})=S_{j}^{(\ell)}$, and $S_j\subseteq S_j^{(\ell)}$. We define $\vecS^{(\ell)} = (S_{\ell}^{(\ell)},\dots,S_{d}^{(\ell)})$.
\end{enumerate}
\end{definition}

Note that $\Sb$ is a concept which we will use to show that a $d$-depth quantum circuit cannot successfully evaluate $f_d^*$ with high probability. Hence, we will choose $\Sb$ in the ways such that some properties are satisfied depending on the computational models we are considering. We will see how to construct $\Sb$ in the following sections.

With the concept of $\Sb$, we can introduce the notation \emph{Shadow} which we will use to analyze the shuffling oracle.
\begin{definition}[Shadow function]
Let $\vecf:=(f_0,\dots,f_d)$ be a $(d,f)$-shuffling of $f: \Z_2^n\rightarrow \Z_2^n$. Fix the hidden sets $\Sb:= (\vecS^{(0)},\dots,\vecS^{(d)})$. The shadow $\vecg$ of $\vecf$ in $\vecS^{(\ell)} = (S_{\ell}^{(\ell)},\dots,S_d^{(\ell)})$ is as follows: For $j=\ell,\dots,d$, let $g_j$ be the function such that if $x\in S_{j}^{(\ell)}$, $g_j(x)=\bot$; otherwise, $g_j(x) = f_j(x)$. We let $\vecg:= (f_0,\dots,f_{\ell-1},g_{\ell},\dots,g_d)$.
\end{definition}

We can also represent a $(d,f)$-shuffling $\vecf$ in terms of mappings corresponding to elements in $\vecS^{(\ell)}$ of $\Sb$. 

\begin{definition}[$\vecf^{(\ell)}$ and $\vecf^{(\ell)}$]\label{def:shaving_f}
For $\ell=1,\dots,d$, we let $f_j^{(\ell)}$ be $f_j$ on $S_{j}^{(\ell-1)}\setminus S_{j}^{(\ell)}$ and $\hat{f}_j^{(\ell)}$ be $f_j$ on $S_{j}^{(\ell)}$. Then, we define 
\begin{align*}
    &\vecf^{(1)} := (f_0,f_1^{(1)},\dots,f_d^{(1)})\quad  \quad\vecf^{(d+1)} := (\hat{f}_d^{(d)}),\,\mbox{and}\\
    & \vecf^{(\ell)} := (\hat{f}^{(\ell)}_{\ell-1},f^{(\ell)}_{\ell},\dots,f_{d}^{(\ell)})
\end{align*}
for $\ell=2,\dots,d$. Also, we define 
\begin{align*}
    &\nvecf^{(0)} := (f_0,\dots,f_d)\,\mbox{and}\\
    &\nvecf^{(\ell)} := (\hat{f}_{\ell}^{(\ell)},\dots,\hat{f}_{d}^{(\ell)})
\end{align*}
for $\ell = 1,\dots,d$.
\end{definition}

We can say that $\nvecf^{(\ell)}$ is the mapping of elements in $\vecS^{(\ell)}$,  and $\vecf^{(1)},\dots,\vecf^{(\ell)}$ are the mappings out of $\vecS^{(\ell)}$ for $\ell\in [d]$.

We rewrite $\vecf$ in the following form 
\[
    \vecf:=(\vecf^{(1)},\dots,\vecf^{(d+1)}).
\]
This representation will be convenient for our analysis. There are several facts corresponding to this representation.
\begin{observation}\label{ob:1}
Let $\vecf$ be an uniform $(d,f)$-shuffling of some function $f$.  
\begin{itemize}
    \item The key function $f^*_d$ is in $\vecf^{(d+1)}:= (\hat{f}_d^{(d)})$.
    \item Let $\vecg:= (f_1,\dots,f_{\ell-1},g_{\ell},\dots,g_{d})$ be the shadow of $\vecf$ in $\vecS^{(\ell)}$, $\vecg$ must be consistent with $\vecf$ on $\vecf^{(1)},\dots,\vecf^{(\ell)}$, and maps the elements in $\vecS^{(\ell)}$ to $\bot$. In other words, $\hat{\vecf}^{(\ell)}$ is totally blocked by $\bot$ when given $\vecg$.  
    \item For $\ell = 1,\dots,d$, conditioned on $\vecf^{(1)},\dots,\vecf^{(\ell)}$, the function $\hat{f}^{(\ell)}_{j}$ is still drawn uniformly randomly from $\mathcal{P}(S_{j}^{(\ell)},S_{j+1}^{(\ell)})$ for $j=\ell,\dots,d-1$ according to the definition of $\D(f,d)$.
\end{itemize}
\end{observation}

\begin{remark}\label{remark:uncorrelation}
In this work, we will say a quantum state $\rho$ or a classical bit string $\bar{s}$ is \textbf{uncorrelated} to $\vecf^{(\ell)}$. This means that if we replace $\vecf^{(\ell)}$ by any other function, the process which outputs $\rho$ or $\bar{s}$ will not change the output distribution. On the other hand, if a quantum state $\rho$ or a bit string $\bar{s}$ is correlated to $\vecf^{(\ell)}$, then we will assume a process which is given $\rho$ (or $\bar{s}$) knows everything about $\vecf^{(\ell)}$ without loss of generality.     
\end{remark}

\subsection{Semi-classical Shuffling Oracle}
In this section, we are going to combine the concepts of ``semi-classical'' oracle introduced in~\cite{AHU18} and the shuffling oracle together.

\begin{definition}[$U^{\vecf\setminus \vecS^{(\ell)}}$ ]~\label{def:semi_oracle_2}
Let $f$ be an arbitrary function and $\vecf$ be a random $(d,f)$-shuffling of $f$. Let $\Sb:= (\vecS^{(0)},\dots,\vecS^{(d)})$ be a sequence of hidden sets.  Let $U$ be single-depth quantum circuit. For $\ell\in [d]$, we define $U^{\vecf\setminus \vecS^{(\ell)}}$ to be an unitary operating on registers $(\regR, \regI)$ where $\regI$ is a single-qubit register.  $U^{\vecf\setminus \vecS^{(\ell)}}$ simulates $\vecf U$ and that: 
\begin{itemize}
    \item[] Before applying $\vecf$, $U^{\vecf\setminus \vecS^{(\ell)}}$ first applies $U_{\vecS^{(\ell)}}$ on $(\regR,\regI)$ and then performs $\vecf$. Here $U_{\vecS^{(\ell)}}$ is defined by: 
    \[
        U_{\vecS^{(\ell)}}\ket{(\ell,\X_{\ell}),\dots,(d,\X_d)}_{\regR}\ket{b}_{\regI}:=
        \left\{
            \begin{array}{ll}
            U_{\vecS^{(\ell)}}\ket{(\ell,\X_{\ell}),\dots,(d,\X_d)}\ket{b} \mbox{ if every }\X_i\cap S_i^{(\ell)}=\phi,&\\
            U_{\vecS^{(\ell)}}\ket{(\ell,\X_{\ell}),\dots,(d,\X_d)}\ket{b+1\mod 2} \mbox{ otherwise.}&
        \end{array}
        \right.
    \]
\end{itemize}

\end{definition}

In other words, for any state $\ket{\psi}_{\regR,\regI}$,
\begin{align*}
    & U^{\vecf\setminus \vecS^{(\ell)}}\ket{\psi} := \vecf U_{\vecS^{(\ell)}} U \ket{\psi}
\end{align*}

In the following, we define a quantity which is the probability that the parallel queries are in a particular hidden set $\vecS^{(\ell)}$.
\begin{definition}[$\Pr(find\; \vecS^{(k+1)}: U^{\vecf\setminus\vecS^{(k+1)}},\rho)$]\label{def:find_prob_2}
Let $k,d\in \mathbb{N}$ and $k+1<d$.  Let $U$ be a single-depth quantum circuit and $\rho$ be any input state. We define 
\[
    \Pr[find\; \vecS^{(k+1)}: U^{\vecf\setminus \vecS^{(k+1)}},\rho]:=\E\left[\Tr\left(\left(I_{\regR}\otimes(I-\opro{0}{0})\right)_{\regI}\circ U^{\vecf\setminus \vecS^{(k+1)}}\circ\rho\right)\right].
\]
\end{definition}

Following Def.~\ref{def:find_prob_2}, we let $\ket{\psi}$ be a pure state and $U$ be a single-depth quantum circuit. Then, 
\[
U^{\vecf\setminus \vecS^{(\ell)}}\ket{\psi}_{\regR}\ket{0}_{\regI} :=\ket{\phi_{0}}_{\regR}\ket{0}_{\regI} +  \ket{\phi_{1}}_{\regR}\ket{1}_{\regI} 
\]
and $\Pr[find\; \vecS^{(\ell)}: U^{\vecf\setminus \vecS^{(\ell)}},\ket{\psi}] = \E[\|\ket{\phi_{1}}_{\regR}\|^2]$. Note that $\ket{\phi_0}$ and $\ket{\phi_{1}}$ are orthogonal by the fact that $\ket{\phi_0}$ involves no query to $\vecS^{(\ell)}$ but $\ket{\psi_1}$ does. Therefore, 
\begin{eqnarray*}
    \vecf U\ket{\psi} = \ket{\phi_{0}} +  \ket{\phi_{1}}.\label{eq:o2h_core}
\end{eqnarray*}

\subsection{Oneway-to-hiding (O2H) lemma for the shuffling oracle}

Here, we extend the Oneway-to-hiding lemma (O2H lemma) by Ambainis et al.~\cite{AHU18} to the shuffling oracle. Briefly, the O2H lemma shows that for any two functions $g$ and $h$, the probability for a quantum algorithm to distinguish them is bounded by the probability that the quantum algorithm ever "find" an element in the input domain that $g$ and $h$ disagree with each other times the depth of the quantum algorithm.

\begin{lemma}[O2H lemma for the shuffling oracle]\label{lem:o2h}
Let $k,d\in \mathbb{N}$ satisfying that $k<d$. Let $U$ be any single depth quantum circuit and $\rho$ be any input state. Let $f$ be any function from $\Z_2^n$ to $\Z_2^{n}$. Let $\vecf$ be a random $(d,f)$-shuffling of $f$. Let $\Sb:=(\vecS^{(0)},\dots,\vecS^{(d)})$ be a sequence of random hidden sets as defined in Def.~\ref{def:S}. Let $\vecg$ be the shadow of $\vecf$ in $\vecS^{(k)}$. Then, for any binary string $t$,
\begin{eqnarray*}
    |\Pr[\Pi_{0/1}\circ \vecf  U(\rho) = t] - \Pr[\Pi_{0/1}\circ \vecg  U(\rho) = t]|&\leq& B(\vecf U(\rho),\vecg  U(\rho)) \\
    &\leq& \sqrt{2\Pr[find\;\vecS^{(k)}: U^{\vecf\setminus \vecS^{(k)}},\rho]}, 
\end{eqnarray*}
where $\Pi_{0/1}$ is the measurement in the standard basis. Here the probability is over $\vecf$, $\Sb$, and the randomness of the quantum mechanism. 
\end{lemma}

\begin{proof}
We will prove the case where the initial state is a pure state and then the general case directly follows from the concavity of the mixed state.  For simplicity, we denote 
$\Pr[find\;\vecS^{(k)}: U^{\vecf\setminus \vecS^{(k)}},\rho]$ as $P_{find}$.

Fix $\vecf$ and $\vecS^{(k)}$. We let $\ket{\psi}$ be any initial state and 
\[
    \vecf U_{\vecS^{(k)}} U \ket{\psi}_{\regR}\ket{0}_{\regI}:= \ket{\phi_0}_{\regR}\ket{0}_{\regI} + \ket{\phi_1}_{\regR}\ket{1}_{\regI}, 
\]
where $\ket{\phi_0}$ and $\ket{\phi_1}$ are two unnormalized states. $\ket{\phi_0}$ and $\ket{\phi_1}$ are orthogonal due to the fact that all queries $\ket{\phi_0}$ consists of are not in $\vecS^{(k)}$, while the queries $\ket{\phi_1}$ performs are elements in $U_{\vecS^{(k)}}$. This, therefore, implies that 
\[
\ket{\psi_f}:=\vecf U\ket{\psi} = \ket{\phi_0}+\ket{\phi_1}\mbox{ and}
\]
Similarly, we let 
\[
    \vecg UU_{\vecS^{(k)}}\ket{\psi}_{\regR}\ket{0}_{\regI}:= \ket{\phi_0}_{\regR}\ket{0}_{\regI} + \ket{\phi^{\bot}_1}_{\regR}\ket{1}_{\regI}, 
\]
and due to the same fact that $\ket{\phi_0}$ and $\ket{\phi^{\bot}_1}$ are orthogonal, we have that
\[
\ket{\psi_g}:= \vecg U\ket{\psi} =\ket{\phi_0}+\ket{\phi^{\bot}_1}.
\]
Here, $\ket{\phi^{\bot}_1}$ and $\ket{\phi_1}$ are orthogonal since $\vecg$ maps all elements in $\vecS^{(k)}$ to $\bot$. 

Let $P_{find}(\vecf,\vecS^{(k)})$ be the probability that a standard-basis measurement in the register $\regI$ of $\vecf U_{\vecS^{(k)}}U\ket{\psi}_{\regR}\ket{0}_{\regI}$ returns $1$, which is equal to $\|\ket{\phi_1}\|^2$. Consider the two-norm distance between $\ket{\psi_f}$ and $\ket{\psi_g}$,
\begin{eqnarray*}
    \|\ket{\psi_f}-\ket{\psi_g}\|^2 &=& \|\ket{\phi_1}-\ket{\phi_1^{\bot}}\|^2 \\
    &=&\|\ket{\phi_1}\|^2 + \|\ket{\phi_1^{\bot}}\|^2\\
    &\leq& 2\|\ket{\phi_1}\|^2 = 2P_{find}(\vecf,\vecS^{(k)}).
\end{eqnarray*}
The second equality follows from the fact that $\ket{\phi_1}$ and $\ket{\phi^{\bot}_1}$ are orthogonal. The inequality is because that $\|\ket{\phi_1}\|^2 = \|\ket{\phi^{\bot}_1}\|^2 = 1-\|\ket{\phi_0}\|^2$. 

Then, consider the case that $\vecS^{(k)}$ and $\vecf$ are random. The output states of $\vecf U$ and $\vecg U$ becomes  
\begin{align*}
    &\rho_f := \sum_{\vecf,\vecS^{(k)}} \Pr[\vecf\wedge \vecS^{(k)}] \opro{\psi_f}{\psi_f},\mbox{ and }\\ 
    &\rho_g := \sum_{\vecf,\vecS^{(k)}} \Pr[\vecf\wedge \vecS^{(k)}] \opro{\psi_g}{\psi_g}.
\end{align*}

Consider the fidelity of these two mixed states. 
\begin{eqnarray*}
    F(\rho_f,\rho_g) &&=F\left(\sum_{\vecf,\vecS^{(k)}} \Pr[\vecf\wedge \vecS^{(k)}]\opro{\psi_f}{\psi_f},\, \sum_{\vecf,\vecS^{(k)}} \Pr[\vecf\wedge \vecS^{(k)}]\opro{\psi_g}{\psi_g}   \right)\\
    &&\geq \sum_{\vecf,\vecS^{(k)}} \Pr[\vecf\wedge \vecS^{(k)}] F\left(\opro{\psi_f}{\psi_f}, \opro{\psi_g}{\psi_g}   \right)\\
    &&\geq 1-\frac{1}{2}\cdot \sum_{\vecf,\vecS^{(k)}} \Pr[\vecf\wedge \vecS^{(k)}] \|\ket{\psi_f}-\ket{\psi_g}\|^2 \\
    && \geq 1-\frac{1}{2}\cdot \sum_{\vecf,\vecS^{(k)}} \Pr[\vecf\wedge \vecS^{(k)}]2P_{find}(\vecf,\vecS^{(k)}). \\
    &&\geq 1- P_{find}. 
\end{eqnarray*}

Then, 
\begin{eqnarray*}
B(\rho_f,\rho_g) &=& \sqrt{2 - 2F(\rho_f,\rho_g)}\\
&\leq& \sqrt{2-2(1-P_{find})} = \sqrt{2P_{find}}. 
\end{eqnarray*}

Finally, by Lemma~\ref{lem:prob_to_norm}, 
\[
    |\Pr[\Pi_{0/1}\circ \vecf  U(\rho) = t] - \Pr[\Pi_{0/1}\circ \vecg  U(\rho) = t]| \leq \sqrt{2P_{find}}. 
\]

\end{proof}


\subsection{Bounding the finding probability}

As we have just shown that the probability of distinguishing $\vecf$ and its shadow can be bounded by the probability of finding the ``shadow''. Then, we would like to show how to bound the finding probability. 

Follow the previous section, we let $\vecf$ be a random ($d,f$)-shuffling of $f$ and $\Sb:=(\vecS^{(0)},\dots,\vecS^{(d)})$ be a sequence of random hidden sets as defined in Definition~\ref{def:S} (which could be chosen according to arbitrary distribution). We show that the finding probability of $\vecS^{(k)}$ is bounded.  
\begin{lemma}\label{lem:find}
Suppose $\Pr[x\in S_i^{(k)}|x\in S_i^{(k-1)}] \leq p$ for $i=k,\dots,d$. Then for any single-depth quantum circuit $U$ and initial state $\rho$, which are promised to be uncorrelated to $\hat{\vecf}^{(k-1)}$ and $\vecS^{(k)}$,\footnote{Here, $\rho$ and $U$ can be arbitrarily correlated to $\vecf^{(1)},\dots,\vecf^{(k-1)}$, and $\vecf$ and $\vecS^{(k)}$can be sampled arbitrarily.}
\[
    \Pr[find\; \vecS^{(k)}: U^{\vecf\setminus \vecS^{(k)}},\rho]\leq q\cdot p, 
\]
where $q$ is the number of queries $U$ performs.
\end{lemma}
\begin{proof}
It is sufficient to prove the case where $\rho$ is a pure state. Let $\ket{\psi}$ be the initial state and be promised to be uncorrelated to $\hat{\vecf}^{(k-1)}$. We represent $\vecf  U_{\vecS^{(k)}}U\ket{\psi}$ as 
\[
    \sum_{\X_0,\dots,\X_d} c(\X_0,\dots,\X_d) \left(\bigotimes_{i=0}^d\ket{i,\X_i,f_i(\X_i)}\right)\ket{w(\X_0,\dots,\X_d)}\ket{b(\X_1,\dots,\X_d)}_{\regI}
\]
where $b(\X_1,\dots,\X_d)=1$ if there exists $i\in[\ell,d]$ such that $\X_i\cap S_i^{(k)}\neq \{\phi\}$; otherwise, $b(\X_1,\dots,\X_d)=0$. We can assume all queries are in $\vecS^{(k-1)}$ without loss of generality. Since both $U$ and $\ket{\psi}$ are uncorrelated to $\vecS^{(k)}$, the probability that $b(\X_1,\dots,\X_d)= 1$ is at most $p\cdot(\sum_{i=k}^d |\X_i|)$ for all $\X_1,\dots,\X_d$ by union bound. Therefore,
\begin{eqnarray*}
    &&\Pr[find\; \vecS^{(k)}: U^{\vecf\setminus \vecS^{(k)}},\ket{\psi}]\\ 
    &&= \E\left[\left\|\sum_{\X_0,\dots,\X_d: b(\X_0,\dots,\X_d) =1}c(\X_0,\dots,\X_d) \left(\bigotimes_{i=0}^d\ket{i,\X_i,f_i(\X_i)}\right)\ket{w(\X_0,\dots,\X_d)} \right\|^2\right]\\
    &&= \sum_{\X_0,\dots,\X_d: b(\X_0,\dots,\X_d) =1} |c(\X_0,\dots,\X_d)|^2 \cdot \Pr[\bigvee_{i=k}^d (\X_i\cap S_i^{(k)} \neq \{\phi\})]\\
    &&\leq q\cdot p
\end{eqnarray*}
for $q$ the number of queries $U$ performs. The second equality follows from the fact that for different set of queries, $\ket{i,\X_i,f_i(\X_i)}$'s are orthogonal. The last inequality follows from the union bound.     
\end{proof}

%% file: QCd.tex
\section{The \texorpdfstring{$\SSP{d}$}{Lg} is hard for \texorpdfstring{$\class{QNC_d}$}{Lg}}\label{sec:qcd}
We start from showing that the $\SSP{d}$ is intractable for any $\class{QNC_d}$ circuit as a warm-up.  We first prove the main theorem in this section.
\begin{theorem}\label{thm:qcd}
Let $n,d\in \mathbb{N}$. Let $(\A,\rho)$ be any $d$-depth quantum circuit and initial state. Let $f$ be a random Simon's function from $\Z_2^n$ to $\Z_2^{n}$ with hidden shift $s$. Give $\A$ the access to the shuffling oracle $\Oru$. Let $\vecf$ be the $(d,f)$-shuffling sampled from $\Oru$, then
\begin{align*}
    &\Pr[\A^{\vecf}(\rho) = s] \leq d\cdot\sqrt{\frac{ \poly(n)}{2^n}}+\frac{1}{2^n}.
\end{align*}
    
\end{theorem}

\begin{proof}
We choose $\Sb=(\vecS^{(0)},\dots,\vecS^{(d)})$ according to Procedure~\ref{fig:S_1} and represent $\vecf$ in form $(\vecf^{(1)},\dots,\vecf^{(d+1)})$ regarding to $\Sb$. Let $\vecg_{\ell}$ be the shadow function of $\vecf$ in $\vecS^{(\ell)}$ for $\ell\in[d]$. We define
\[
    \A^{\vecg} := U_{d+1}\vecg_{d} U_d\cdots \vecg_{1}U_1.
\]
Then, for any initial state $\rho_0$ which is uncorrelated to $\vecf$, 
\begin{eqnarray*}
&&|\Pr[\A^{\vecf}(\rho_0)=s] - \Pr[\A^{\vecg}(\rho_0)=s]|\\
&&= |\Pr[U_{d+1}\vecf U_d\cdots \vecf  U_1(\rho_0)=s]-\Pr[U_{d+1} \vecg_{d}   U_d\cdots \vecg_{1}  U_1(\rho_0)=s]|\\
&&\leq |\Pr[U_{d+1}  \vecf U_d\cdots \vecf  U_2\vecf U_1(\rho_0)=s] - \Pr[U_{d+1}  \vecf U_d\cdots \vecf  U_2\vecg_1U_1(\rho_0)=s]|\\ &&+|\Pr[U_{d+1}  \vecf U_d\cdots \vecf  U_2\vecg_1 U_1(\rho_0)=s] - \Pr[U_{d+1}  \vecg_d U_d\cdots \vecg_2  U_2\vecg_1U_1(\rho_0)=s]|\\
&&\leq \sum_{i=1}^{d}B(\vecf U_i(\rho_{i-1}),\rho_i)\\
&&\leq \sum_{i=1}^{d} \sqrt{2\Pr[find\; \vecS^{(i)}:U_i^{\vecf\setminus \vecS^{(i)}},\rho_{i-1}]}
\end{eqnarray*}
where $\rho_i:= \vecg_iU_i \rho_{i-1} U_i^{\dag}\vecg_i^{\dag}$ for $i\geq 1$.

\floatname{algorithm}{Procedure}
\begin{algorithm}[h]
    \begin{mdframed}[style=figstyle,innerleftmargin=10pt,innerrightmargin=10pt]
    Let $d,n\in \mathbb{N}$ and $f$ a random Simon's problem from $\Z_2^n\rightarrow \Z_2^{n}$. Consider $\vecf\sim\D(f,d)$, we construct $\Sb$ as follows:  
    \begin{itemize}
    \item Let $\vec{S}^{(0)}:=(S_0^{(0)},\dots,S_{d}^{(0)})$, where $S_j^{(0)} := \Z_2^{(d+2)n}$ for $j=0,\dots,d$.  
    \item For $\ell=1,\dots,d$, 
    \begin{enumerate}
        \item let $S_\ell^{(\ell)}$ be a subset chosen uniformly at random with the promise that $|S_\ell^{(\ell)}|/|S_\ell^{(\ell-1)}| =\frac{1}{2^n}$ and $S_\ell\subset S_\ell^{(\ell)}$;
        \item for $j=\ell+1,\dots, d$,  let $S_j^{(\ell)}:= \{f_{j-1}\circ\cdots\circ f_{\ell}(S_{\ell}^{(\ell)})\}$; 
        \item let $\vecS^{(\ell)} := (S_{\ell}^{(\ell)},\dots,S_{d}^{(\ell)})$.
    \end{enumerate}
    \item We then let $\Sb := (\vecS^{(0)},\dots,\vecS^{(d)})$.
    \end{itemize}
    \caption{The hidden sets for $\class{QNC_d}$}
    \label{fig:S_1}
    \end{mdframed}
\end{algorithm}

It is not hard to see that $ \Pr[\A^{\vecg}(\rho_0)=s]$ is at most $\frac{1}{2^n}$. This follows from the fact that 
$\vecg_1,\dots,\vecg_d$ does not contain information of $f^*_d$ and therefore $\A^{\vecg}(\rho_0)$ can do no better than guess. The rest to show is that $\Pr[find\; \vecS^{(i)}:U_i^{\vecf\setminus \vecS^{(i)}},\rho_{i-1}]$ is at most $\frac{\poly(n)}{2^n}$ for all $i\in [d]$. To prove it, we show that $\Pr[x\in S_j^{(\ell)}|x\in S_j^{(\ell-1)}]= \frac{1}{2^n}$ for $\ell=1,\dots,d$ and $j=\ell,\dots,d$. We prove it by induction on $\ell$. 

For the base case $\ell=1$, for all $j\in [d]$, and $x\in S_j^{(0)}$, 
\begin{eqnarray*}
\Pr[x\in S_{j}^{(1)}] &=& \Pr[x\in S_j]\Pr[x\in S_{j}^{(1)}| x\in S_j] + \Pr[x\notin S_j]\Pr[x\in S_{j}^{(1)}| x\notin S_j]\nonumber\\
&=&\Pr[x\in S_j] + (1-\Pr[x\in S_j])\Pr[x\in S_{j}^{(1)}| x\notin S_j]\nonumber\\
&=& \frac{1}{2^{(d+1)n}} + (1-\frac{1}{2^{(d+1)n}}) \frac{2^{(d+1)n}-2^n}{2^{(d+2)n}-2^n} = \frac{1}{2^n}.   
\end{eqnarray*}
The second equality is because that $x\in S_j$ implies $x\in S_j^{(1)}$ and the third inequality follows from the fact that $f_0,\dots,f_{d-1}$ are uniformly random one-to-one functions.

Suppose that the randomness holds for $\ell=k$. Note that $\rho_{k}$ could be correlated to $\vecf^{(1)},\dots,\vecf^{(k)}$, and therefore, we assume $\vecf^{(1)},\dots,\vecf^{(k)}$ are given to $\A$ without loss of generality as in Remark~\ref{remark:uncorrelation}. Then, given $\vecf^{(1)},\dots,\vecf^{(k)}$, for $j\geq k+1$ and $x\in S_j^{(k)}$ 
\begin{eqnarray*}
    \Pr[x\in S_{j}^{(k+1)}] &=& \Pr[x\in S_j]\Pr[x\in S_j^{(k)}| x\in S_j]\nonumber\\ 
    &&+\Pr[x\notin S_j]\Pr[x\in S_j^{(k)}| x\notin S_j]\nonumber\\
    &=& \Pr[x\in S_j] + (1-\Pr[x\in S_j])\Pr[x\in S_j^{(k)}| x\notin S_j]\nonumber\\ 
    &=& \frac{2^n}{2^{(d+2-k)n}} + (1-\frac{2^n}{2^{(d+2-k)n}}) \frac{2^{(d+1-k)n}-2^n}{2^{(d+2-k)n}-2^n} \\
    &=& \frac{1}{2^n}.\nonumber
\end{eqnarray*}
The second last equality follows from the fact that given $\vecf^{(1)},\dots,\vecf^{(k)}$, $\hat{f}_j^{(k)}$ is still an uniformly random one-to-one function for $j=k,\dots,d-1$.

Finally, $U_i$ and $\rho_{i-1}$ are uncorrelated to $\hat{\vecf}^{(i)}$.  
By Lemma~\ref{lem:find}, $\Pr[find\; \vecS^{(i)}:U_i^{\vecf\setminus \vecS^{(i)}},\rho_{i-1}]$ is at most $q_i\cdot \frac{1}{2^n}$ where $q_i$ is the number of queries $U_i$ performs. Therefore, 
\begin{eqnarray*}
    &&|\Pr[\A^{\vecf}(\rho_0)=s] - \Pr[\A^{\vecg}(\rho_0)=s]|\\
    && \leq\sum_{i=1}^{d} \sqrt{2\Pr[find\; \vecS^{(i)}:U_i^{\vecf\setminus \vecS^{(i)}},\rho_{i-1}]}\\ 
    &&\leq \sum_{i=1}^d \sqrt{\frac{2q_i}{2^n}} = d\cdot\sqrt{\frac{\poly(n)}{2^n}}. 
\end{eqnarray*}
\end{proof}

Theorem~\ref{thm:qcd} shows that the search $\SSP{d}$ is hard for any $\class{QNC_d}$ circuit. By following the same proof, we can show that for any $\class{QNC_d}$ circuit, the $\SSP{d}$ is also hard.  
\begin{theorem}
The $\SSP{d}$ cannot be decided by any $\class{QNC_d}$ circuit with probability greater than $\frac{1}{2}+d\cdot \sqrt{\frac{\poly(n)}{2^n}}$.  
\end{theorem}
\begin{proof}
We also consider the same shadow $\vecg$ in the proof of Theorem~\ref{thm:qcd}. Following that proof, for any $\rho$ and $\A$, 
\[
    |\Pr[\A^{\vecf}(\rho_0)=0] - \Pr[\A^{\vecg}(\rho_0)=0]|\leq d\cdot\sqrt{\frac{\poly(n)}{2^n}}. 
\]
Then, the rest to check is that $\A^{\vecg}$ cannot solve the $\SSP{d}$ with high probability. Similar to the case of the search $\SSP{d}$, since $\vecg_1,\dots,\vecg_d$ have the core function $f^*_d$ be blocked,  $\A^{\vecg}$ has no information about $f$ and thus cannot do better than guess. This implies that $\Pr[\A^{\vecf}(\rho_0)=0] \leq \frac{1}{2} + d\cdot\sqrt{\frac{\poly(n)}{2^n}} < 2/3$. 

\end{proof}

Therefore, the language defined in Def.~\ref{def:language} is also hard for $\class{QNC}$ circuit. 
\begin{corollary}
Let $\Ora$ and $\Lang(\Ora)$ be defined as in Def.~\ref{def:language}. $\Lang(\Ora)\notin \class{BQNC}^{\Ora}$
\end{corollary}

%% file: QCd_BPP.tex
\section{The \texorpdfstring{$\SSP{d}$}{Lg} is hard for \texorpdfstring{$d$}{Lg}-QC scheme}\label{sec:qcd_bpp}

The main theorem we are going to show in this section is that the search $\SSP{d}$ is hard for all $d$-QC scheme. 
\begin{theorem}\label{thm:qcd_bpp}
Let $d,n\in \mathbb{N}$. For any $d$-QC scheme $\A$ and initial state $\rho$, let $f$ be a random Simon's function from $\Z_2^N$ to $\Z_2^{n}$ with hidden shift $s$, and $\vecf\sim\D(f,d)$, then
\[
    \Pr[\A^{\vecf}(\rho) = s] \leq d\cdot\sqrt{\frac{ \poly(n)}{2^n}}.
\] 
\end{theorem}

Before proving Theorem~\ref{thm:qcd_bpp}, we first recall the classical lower bound for the Simon's problem. However, for the purpose of 
\begin{lemma}\label{lem:simon_classical_bound}
Let $\A_c$ be any PPT algorithm. Let $f: \Z_2^n\rightarrow \Z_2^n$ be a uniformly random Simon function. Let $q\in \mathbb{N}$ be the number of queries $\A_c$ performs, and $S\subset \Z_n$ be the set where $f(x)$ is known for $x\in S$ and $f(x)\neq f(x')$ for $x\neq x'$. Then the probability that $\A_c^{f}(S,f(S))$ outputs the hidden shift correctly is at most 
\[
    \frac{(q+1+|S|)^2}{2^{n+1}-(q+1+|S|)^2},
\]
which is $O(\poly(n)/2^n)$ when $q$ and $|S|$ are polynomial in $n$. 
\end{lemma}
\begin{proof}
We only need to consider the case where $\A_c$ is a deterministic algorithm. A probabilistic algorithm can be seen as a convex combination of deterministic algorithm; therefore, the success probability of a probabilistic algorithm must be an average over deterministic algorithms. 

Let $S\subset \Z_2^n$ and $f(x)\neq f(y)$ for $x,y\in S$ and $x\neq y$. The probability that $\A_c$ finds a collision is 
\[
    \Pr[collision: \A_c^f(S)] \leq \sum_{i=1}^q \frac{i+|S|-1}{2^n-(i+|S|)^2/2} \leq \frac{(q+|S|)^2}{2^{n+1}-(q+|S|)^2}.  
\]

For any algorithm which can find $s$ with probability $p$ by performing $q$ queries, it can find a collision with the same probability by performing $q+1$ queries. The probability of finding a collision by using $q+1$ queries is at most $\frac{(q+|S|+1)^2}{2^{n+1}-(q+|S|+1)^2}$. Therefore, 
\[
    \Pr[\A_c^f(S) = s] \leq \frac{(q+|S|+1)^2}{2^{n+1}-(q+|S|+1)^2}. 
\]

\end{proof}

\subsection{Proof of Theorem~\ref{thm:qcd_bpp}}

Recall that we can represent a $d$-QC scheme $\A$ with access to $\vecf\sim \D(f,d)$ as
\[
    \A_c^{\vecf}\circ(\Pi_{0/1}\circ \vecf  U_d \circ \A_c^{\vecf})\circ\cdots\circ(\Pi_{0/1}\circ \vecf  U_1 \circ \A_c^{\vecf}) .
\]
We denote $\Pi_{0/1}\circ \vecf  U_i \circ \A_c^{\vecf}$ as $L_i^{\vecf}$ for $i=1,\dots,d$ and rewrite the representation above as $\A_c^{\vecf}\circ L_d^{\vecf}\circ\cdots\circ L_1^{\vecf}$.
We let $q_i$ be the number of quantum queries and $r_i$ be the number of classical queries the algorithm performs in $L_i$. We let $q:=\sum_{i=1}^{d}q_i$ and $r:=\sum_{i=1}^dr_i$.


For the ease of the analysis, we allow $\A_c$ to learn the whole path from $f_0$ to $f_d$ by just one query, which we called the ``\textbf{path query}''. It is worth noting that $\A_c$ that can make path queries can be simulated by the original model. This follows from the fact that the original model can achieve the same thing by using $d$ times as many queries as $\A_c$.

To prove the theorem, we need to define an $\Sb$ which has property that $\Pr[x\in S_j^{(\ell)}]\leq p$ as described in Lemma~\ref{lem:find} for some $p$ that is small enough. In the following, we show that $\Sb=(\vecS^{(1)},\dots,\vecS^{(d)})$ in Procedure~\ref{fig:S_2} satisfies this property with $p=\frac{1}{2^n}$.

\begin{claim}\label{claim:independence_2}
Let $d,n\in\mathbb{N}$. Let $\ell\in [d]$. Let $\A_c$ be any randomized algorithm. Let $f:\Z_2^n\rightarrow\Z_2^{n}$ be any function.  Let $\vecf\sim\D(f,d)$. Let $\vecS^{(1)},\dots,\vecS^{(\ell)}$ and $\vecT^{(1)},\dots,\vecT^{(\ell)}$ be be defined as in Procedure~\ref{fig:S_2} regarding to $\A_c$. Given $(\vecf^{(1)},\dots,\vecf^{(\ell-1)})$ and $(\vecT^{(1)},\dots,\vecT^{(\ell)})$, then 
\[  
\Pr_{\vecf,\Sb}[x\in S_j^{(\ell)}|x\in S_j^{(\ell-1)}\setminus T_j^{(\ell)}]=\frac{1}{2^n}\mbox{ for } j=\ell,\dots,d.
\]
\end{claim}
\begin{proof}

We prove it via induction on the depth of $\vecf$. For the base case where $\ell=1$, given $\vecT^{(1)}$ and $f_j$ on $T_j^{(1)}$ for $j=1,\dots,d$, for all $i\in [d]$ and $x_i\in S_i^{(0)}\setminus T_i^{(1)}$, 
\begin{eqnarray*}
\Pr[x_i\in S_i^{(1)}] &=& \Pr[x_i\in S_i\setminus T_{i}^{(1)}]+ \Pr[x_i\in S_i^{(1)}\setminus S_i]\\
&=& \frac{2^{n}-|T_i^{(1)}|}{2^{(d+2)n}-|T_i^{(1)}|} + \frac{|S_i^{(1)}|-2^n+|T_i^{(1)}|}{2^{(d+2)n}-|T_i^{(1)}|}\\
&=& \frac{|S_i^{(1)}|}{2^{(d+2)n}-|T_i^{(1)}|} = \frac{1}{2^n}.
\end{eqnarray*}

We now suppose that when $\ell=k-1$, given $\vecf^{(1)},\dots,\vecf^{(k-2)}$ and   $\vecT^{(1)},\dots,\vecT^{(k-1)}$, $$\Pr[x\in S_i^{(k-1)}|x\in S_i^{(k-2)}\setminus T_i^{k-1}] = \frac{1}{2^n}$$ for $i=k-1,\dots,d$.

Then, for $\ell=k$, given $\vecf^{(1)},\dots,\vecf^{(k-1)}$ and $\vecT^{(1)},\dots,\vecT^{(k)}$, for $i=k,\dots,d$ and $x\in S_i^{(k-1)}\setminus T_i^{(k)}$,
\begin{eqnarray*}
\Pr[x\in S_i^{(k)}] &=& \Pr[x\in S_i\setminus (\cup_{m=1}^{k}(T_i^{(m)}))] + \Pr[x\in S_i^{(k)}\setminus S_i]\\
&=& \frac{2^n-\sum_{m=1}^k |T_i^{(m)}|}{|S_i^{(k-1)}| - |T_i^{(k)}|} + \frac{|S_i^{(k)}|-2^n+\sum_{m=1}^k |T_i^{(m)}|}{|S_i^{(k-1)}| - |T_i^{(k)}|}\\
&=& \frac{1}{2^n}.
\end{eqnarray*}

\end{proof}

\floatname{algorithm}{Procedure}
\begin{algorithm}[h]
    \begin{mdframed}[style=figstyle,innerleftmargin=10pt,innerrightmargin=10pt]
    Let $d,n\in \mathbb{N}$ and $f$ a random Simon's problem from $\Z_2^n\rightarrow \Z_2^{n}$. Given $\vecf\sim\D(f,d)$ and $\A$ a $d$-QC scheme, We construct $\Sb$ as follows:   
    \begin{itemize}
    \item Let $\vec{S}^{(0)}:=(S_0^{(0)},\dots,S_{d}^{(0)})$, where $S_j^{(0)} := \Z_2^{(d+2)n}$ for $j=0,\dots,d$.
    \item For $\ell=1,\dots,d$: 
    \begin{enumerate}
        \item After the $\ell$-th $\A_c^{\vecf}$ is applied, let $\vecT^{(\ell)} = (T_{\ell}^{(\ell)},\dots,T_{d}^{(\ell)})$ be the set of points the $\ell$-th $\A_c^{\vecf}$ queried. As we have mentioned before, we allow $\A_c$ to query the whole path by one query. Hence, $f_j(T_j^{(\ell)}) = T_{j+1}^{(\ell)}$ for $j=\ell,\dots,d$.
        \item Let $W_{\ell}^{(\ell-1)} := S_{\ell}^{(\ell-1)}\setminus T_{\ell}^{(\ell)}$. Then, we choose $S_{\ell}^{(\ell)}$ uniformly randomly from $W_{\ell}^{(\ell-1)}$ with the promise that $|S_{\ell}^{(\ell)}|/|W_{\ell}^{(\ell-1)}|=1/2^n$ and $S_{\ell}\setminus (T_{\ell}^{(1)}\cup\cdots\cup T_{\ell}^{(\ell)})\subset S_{\ell}^{(\ell)}$.
        \item For $j=\ell+1,\dots, d$,  let $S_j^{(\ell)}:= \{f_{j-1}\circ\cdots\circ f_{\ell}(S_{\ell}^{(\ell)})\}$. 
        \item Let $\vecS^{(\ell)} = (S_{\ell}^{(\ell)},\dots,S_{d}^{(\ell)})$.
    \end{enumerate}
    \item $\Sb:= (\vecS^{(0)},\dots,\vecS^{(d)})$
    \end{itemize}
    \caption{The hidden sets for $d$-QC scheme}
    \label{fig:S_2}
    \end{mdframed}
\end{algorithm}

Now, we are ready to prove Theorem~\ref{thm:qcd_bpp}.
\begin{proof}[Proof of Theorem~\ref{thm:qcd_bpp}]

We choose the hidden set $\Sb=(\vecS^{(0)},\dots,\vecS^{(d)})$ according to Procedure~\ref{fig:S_2}. In the procedure, we choose $\vecS^{(\ell)}$ after the $\A_c$ in $L_{\ell}$ has performed. We represent $\vecf$ as $(\vecf^{(1)},\dots,\vecf^{(d+1)})$ according to $\Sb$. Let $\vecg_{\ell}$ be the shadow of $\vecf$ in $\vecS^{(\ell)}$ for $\ell\in[d]$. We define
\begin{eqnarray*}
    \A^{\vecg} &&:= \A_c^{\vecf}\circ (\Pi_{0/1}\circ \vecg_dU_d\circ \A_c^{\vecf})\circ\cdots\circ (\Pi_{0/1}\circ \vecg_1U_1\circ \A_c^{\vecf})\\
    &&:=\A_c^{\vecf}\circ L_d^{\vecg_d}\circ\cdots\circ L_1^{\vecg_1}.
\end{eqnarray*}

$\A^{ \vecg}$ succeeds to output the hidden shift with probability at most $\frac{(r +1)^2}{2^n - (r +1)^2}$, where $r$ is the number of queries the classical algorithms perform. Note that the outputs of $U_1,\dots,U_d$ are uncorrelated to $\hat{\vecf}^{(d)}$. This fact implies that given the measurement outcomes of the $i$-th layer quantum unitaries, $\A_c$ can only learn information about $\vecf^{(1)},\dots,\vecf^{(d)}$. This does not give $\A_c$ more information about $f$. Therefore,  $\A_c$ at $L_i$ succeeds with probability at most $\frac{(\sum_{j=1}^{i+1} r_j+1)^2}{2^n-(\sum_{j=1}^{i+1} r_j+1)^2}$ and $\A^{ \vecg}$ succeeds with probability at most $\frac{\poly(n)}{2^n}$. 


Let $\rho_0$ be the initial state and $\rho_i$ be the output state of $(L_i^{\vecg_i}\circ\cdots\circ L_1^{\vecg_1})(\rho_0)$ for $i=1,\dots,d$, we can show that
\begin{eqnarray}
&&|\Pr[\A^{\vecf}(\rho_0)=s] - \Pr[\A^{ \vecg}(\rho_0)=s]| \nonumber\\
&&\leq |\Pr[\A^{\vecf}(\rho_0)=s] - \Pr[(L_d\circ\cdots\circ L_2)^{\vecf}( L_1^{ \vecg}(\rho_0))=s]| +\nonumber\\
&&\quad |\Pr[(L_d\circ\cdots\circ L_2)^{\vecf}(L_1^{ \vecg}(\rho_0))=s] - \Pr[(L_d\circ\cdots\circ L_2)^{ \vecg}(L_1^{ \vecg}(\rho_0))=s]|\nonumber\\
&&\leq B(L_1^{\vecf}(\rho_0),L_1^{ \vecg}(\rho_0)) +\nonumber\\
&&\quad|\Pr[(L_d\circ\cdots\circ L_2)^{\vecf}(\rho_1)=s] - \Pr[(L_d\circ\cdots\circ L_2)^{ \vecg}(\rho_1)=s]|\nonumber\\
&&\leq \sum_{i=1}^{d} B(\rho_i,L_i^{\vecf}(\rho_{i-1})).\label{eq:qcd_bpp_hybrid}\\
&&\leq \sum_{i=1}^{d} \sqrt{2\Pr[find\; \vecS^{(i)}: U_i^{\vecf\setminus \vecS^{(i)}},\rho_{i-1}]} .\label{eq:qcd_bpp_o2h}
\end{eqnarray}
Eq.~\ref{eq:qcd_bpp_hybrid} is by the hybrid argument and Eq.~\ref{eq:qcd_bpp_o2h} is from Lemma~\ref{lem:o2h}.

Then, by Lemma~\ref{lem:find} and Claim~\ref{claim:independence_2}, 
\[
    \Pr[find\; \vecS^{(i)}: U_i^{\vecf\setminus \vecS^{(i)}},\rho_{i-1}] \leq \frac{q_i}{2^n}. 
\]
This implies that 
\[
    \Pr[\A^{\vecf}(\rho_0)=s] \leq \frac{\poly(n)}{2^n} + d\cdot \sqrt{\frac{\poly(n)}{2^n}} \leq \sqrt{\frac{\poly(n)}{2^n}}.
\]
\end{proof}

\subsection{On separating the depth hierarchy of \texorpdfstring{$d$}{Lg}-QC scheme}
By following the same proof of Theorem~\ref{thm:qcd_bpp}, we can show that for any $d$-QC scheme, the $\SSP{d}$ is also hard.  
\begin{theorem}\label{thm:qnc_bpp_2}
The $\SSP{d}$ cannot be decided by any $d$-QC scheme with probability greater than $\frac{1}{2}+\sqrt{\frac{\poly(n)}{2^n}}$.  
\end{theorem}
\begin{proof}
We consider the same shadow $\vecg$ in the proof of Theorem~\ref{thm:qcd_bpp}. Following that proof, for any $\rho_0$ and $\A$, 
\[
    |\Pr[\A^{\vecf}(\rho_0)=1] - \Pr[\A^{\vecg}(\rho_0)=1]|\leq d\cdot\sqrt{\frac{\poly(n)}{2^n}}. 
\]
Then, the rest to check is that $\A^{\vecg}$ cannot solve the $\SSP{d}$ with high probability. In case that $f$ is a random Simon's function, $\A^{\vecg}$ finds $s$ with probability at most $\frac{\poly(n)}{2^n}$. Therefore,  
$\Pr[\A^{\vecg}(\rho_0)=1]$ is at most $1/2+\poly(n)/2^n$. This implies that $\Pr[\A^{\vecf}(\rho_0)=1] \leq \frac{1}{2} + d\cdot\sqrt{\frac{\poly(n)}{2^n}} + \frac{\poly(n)}{2^n} < 2/3$. 

\end{proof}


\begin{corollary}\label{cor:qcd_bpp_2}
For any $d\in \mathbb{N}$, there is a $(2d+1)$-QC scheme which can solve the $\SSP{d}$ with high probability, but there is no $d$-QC scheme which can solve the $\SSP{d}$. 
\end{corollary}
\begin{proof}
This corollary follows from Theorem~\ref{thm:qnc_bpp_2} and Theorem~\ref{thm:dSSP_solve} directly. 
\end{proof}

Finally, we can conclude that  
\begin{corollary}\label{cor:qcd_bpp_3}
Let $\Ora$ and $\Lang(\Ora)$ be defined as in Def.~\ref{def:language}. $\Lang(\Ora)\in \class{BQP}^{\Ora}$ and $\Lang(\Ora)\notin \class{(BQNC^{BPP})}^{\Ora}$. 
\end{corollary}
\begin{proof}
Note that for each $n\in \mathbb{N}$, $\Ora_{unif}^{f_n,d(n)}\in \Ora$ has depth equal to the input size. A quantum circuit with depth $\poly(n)$ can decide if $1^n$ is in $\Lang(\Ora)$ by solving the $\SSP{d}$ by  Theorem~\ref{thm:dSSP_solve}.  However, for $d$-QC scheme which only has quantum depth $d = \poly\log n$, it cannot decide the language by Theorem~\ref{thm:qnc_bpp_2}.
\end{proof}

%% file: BPP_QCd_2.tex
\section{The \texorpdfstring{$\SSP{d}$}{lg} is hard for \texorpdfstring{$d$}{Lg}-CQ scheme}\label{sec:bpp_qcd}

The main result we are going to show here is the following theorem.
\begin{theorem}\label{thm:bpp_qcd}
Let $d,n\in \mathbb{N}$. Let $\A$ be any $d$-CQ scheme. Let $f$ be a random Simon's function from $\Z_2^n$ to $\Z_2^{n}$ with hidden shift $s$. Let $\vecf\sim\D(f,d)$. Then
\[
    \Pr[\A^{\vecf}() = s] \leq d\cdot\sqrt{\frac{ \poly(n)}{2^n}}.
\] 
\end{theorem}

Recall that we can represent a $d$-CQ scheme $\A$ as 
\begin{eqnarray*}
    &&\A^{\vecf}_{c,m+1} \circ (\Pi_{0/1}\circ U_{d+1}\vecf\cdots \vecf U_1\circ \A^{\vecf}_{c,m})\circ\cdots\circ (\Pi_{0/1}\circ U_{d+1}\vecf\cdots \vecf U_1\circ \A^{\vecf}_{c,1})\\
    &&:= \A^{\vecf}_{c,m+1}\circ L^{\vecf}_m\circ\cdots\circ L_1^{\vecf}. 
\end{eqnarray*}
The main difficulty for proving Theorem~\ref{thm:bpp_qcd} is that $L_i^{\vecf}$ can send some short classical advice to the proceeding processes, which advice can be correlated to all mappings in $\vecf$. Therefore, conditioned on the short advice, the distribution of the shuffling oracle may not be uniform enough for us to follow the same proofs above. To deal with this difficulty, we show that given a short classical advice, by fixing the shuffling function on a few paths, the rest paths of the shuffling oracle are still \emph{almost-uniform}.

\subsection{The presampling argument for the shuffling oracle}
Here, we are going to show that for $\vecf\sim\D(f,d)$ and $\bar{a}$ a short classical string correlated to $\vecf$,  we can approximate $\vecf|\bar{a}$ ($\vecf$ given $\bar{a}$) by a \textbf{convex combination} of ($p,1+\delta$)-uniform shuffling functions. In the following, we first define ($p,1+\delta$)-uniform shuffling functions and then prove the statement.

Let $X$ and $Y$ be two sets of elements such that $|X|=|Y|$. Recall that $P(X,Y)$ is the set of all one-to-one functions from $X$ to $Y$. 
\begin{definition}[Random variable $\Hf_{\vec{ X}}$]
Let $k,N\in \mathbb{N}$ and $\vec{ X}:=( X_1,\dots, X_{k+1})$ be a set of sets with size $N$. Let $h_i$ be random variables distributed in $P( X_i, X_{i+1})$ for $i=1,\dots,k$. Then, we define 
\[
    \Hf_{\vec{ X}}:= (h_1,\dots,h_k).
\]
\end{definition}
$\Hf_{\vec{ X}}$ is a sequence of random one-to-one functions which distribution could be arbitrary. In the following, we introduce the distributions we will use shortly.  
 
\begin{definition}[Almost-uniform Shuffling] \label{def:almost_uniform}
Let $k,N,p'\in \mathbb{N}$ and $0<\delta<1$. Let $\vec{X}=( X_1,\dots, X_{k+1})$ be a set of sets with size $N$. Consider $\Hf_{\vec{ X}}=(h_1,\dots,h_k)$. 
\begin{itemize} 
    \item $\Hf_{\vec{ X}}$ is \emph{\textbf{($1+\delta$)-uniform}} if for all subset of paths $\Path = (\vecP_1,\dots,\vecP_m)$ where $$\vecP_i=(x_{i,1},\dots,x_{i,k+1})$$ 
    satisfying $h_j(x_{i,j}) = x_{i,j+1}$ for $i=1,\dots,m$ and $j=1,\dots,k$, 
    \[
        \Pr[\Path\mbox{ is in }\Hf_{\vec{X}}]\leq (1+\delta)^m (\frac{(N-m)!}{N!})^k. 
    \]
    \item $\Hf_{\vec{ X}}$ is \emph{\textbf{$(p',1+\delta)$-uniform}} if there exist a set of paths $\Path'$ with size at most $p'$ such that $\Hf_{\vec{ X}}|\Path'$ is $(1+\delta)$-uniform, i.e., let $\Path'$ be fixed, then for all subset of unfixed paths $\Path = (\vecP_1,\dots,\vecP_m)$, 
    \[
        \Pr[\Path\mbox{ is in }\Hf_{\vec{ X}}|\Path'\mbox{ is in }\Hf_{\vec{ X}}]\leq (1+\delta)^m (\frac{N-m-p')!}{(N-p')!})^k. 
    \]
\end{itemize}
\end{definition}


A convex combination of $(p',1+\delta)$-uniform shuffling functions can be defined by the following formula: 
\begin{eqnarray*}
\Hc:= \sum_{t=1}^{T} p_t \Hf_t,\label{eq:conv_def} 
\end{eqnarray*} 
where $\Hf_1,\dots,\Hf_T$ are ($p',1+\delta$)-uniform shuffling functions and $p_1,\dots,p_T$ are the probabilities for how much each $\Hf_t$ contributes to $\Hc$. In out context, we will only consider convex combination of finitely objects, i.e., $T$ is finite. We will show in Claim~\ref{claim:presamp} that $\vecf$ conditioned on some classical advice $\bar{a}$ is close to a convex combination of finitely many $(p',1+\delta)$-uniform shuffling functions.

\begin{claim}\label{claim:presamp}
Let $0<\gamma,\delta< 1$. Let $p,k,N\in \mathbb{N}$. Let $\vec{ X}=( X_1,\dots, X_{k+1})$ be a set of sets with size $N$. Let $\Hf_{\vec{ X}}=(h_1,\dots,h_k)$ be distributed uniformly. Let $\Path$ be a set of fixed paths on $h_1,\dots,h_k$ and $\bar{a}$ be a $p$-bit advice.  Let $\Hf_{\vec{ X}}|(\Path,\bar{a})$ be $\Hf_{\vec{ X}}$ conditioned on $(\Path,\bar{a})$. Then, 
there exists a convex combination $\Hc|(\Path,\bar{a})$ of ($p',1+\delta$)-uniform shufflings such that
\[
    \Hf_{\vec{ X}}|(\Path,\bar{a}) = \Hc_{\vec{X}}|(\Path,\bar{a}) + \gamma'\Hf'
\]
where $p'\leq \frac{p+\log(1/\gamma)}{\log(1+\delta)}+|\Path|$, $\gamma'\leq \gamma$ and $\Hf'$ is an arbitrary random shuffling.  
\end{claim}

\begin{proof}
We let $\Hf_{\vec{ X}}|(\Path,\bar{a})$ be the random variable conditioned on $\Path$ and $\bar{a}$. It is obvious that conditioned on $\Path$, $\Hf_{\vec{ X}}$ is uniform on the rest. We let $P_1$ be the maximal set of paths satisfying that 
\[
    \Pr[P_1\mbox{ in }\Hf_{\vec{ X}}|(\Path,\bar{a})] \geq (1+\delta)^{|P_1|} (\frac{(N-|P_1|-|\Path|)!}{(N-|\Path|)!})^d. 
\]
Conditioned on $P_1$, we show that $\Hf_{\vec{ X}}$ is $(1+\delta)$-uniform by contradiction.  Suppose there exists another set of paths $P'$ such that 
\[
    \Pr[P'\mbox{ in }\Hf_{\vec{ X}}|(P_1,\Path,\bar{a})]\geq (1+\delta)^{|P'|} (\frac{(N-|P_1|-|P'|-|\Path|)!}{(N-P_1-|\Path|)!})^d.
\]
Then, 
\[
    \Pr[P_1\cup P'\mbox{ in }\Hf_{\vec{ X}}|(\Path,\bar{a})] \geq (1+\delta)^{|P'|+|P_1|} (\frac{(N-|P_1|-|P'|-|\Path|)!}{(N-|\Path|)!})^d
\]
which contradicts the maximallity of $P_1$. This proves that conditioned on $P_1$, $\Hf_{\vec{ X}}|(\Path,\bar{a})$ is $1+\delta$-uniform.

The size of $P_1$ is bounded as follows: Since $\bar{a}$ is a $p$-bit advice, 
\[
    \Pr[P_1\mbox{ in }\Hf_{\vec{ X}}|(\Path,\bar{a})]\leq 2^p (\frac{(N-|\Path|-|P_1|)!}{(N-|\Path|)!})^d
\]
which implies that $|P_1|\leq \frac{p}{\log (1+\delta)}$. 

Now, we can decompose $\Hf_{\vec{ X}}|(\Path,\bar{a})$ as 
\[
    \Pr[P_1\mbox{ in }\Hf_{\vec{ X}}|(\Path,\bar{a})]\cdot \Hf_{\vec{ X}}|(\Path,P_1,\bar{a})+ \Pr[\neg P_1\mbox{ in }\Hf_{\vec{ X}}|(\Path,\bar{a})]\cdot \Hf_{\vec{ X}}|(\Path,\neg P_1,\bar{a}).
\]
Note that $\Hf_{\vec{ X}}|(\Path,\neg P_1,\bar{a})$ may not be $(p',1+\delta)$-uniform. In case that it is not and 
\[
\Pr[\neg P_1\mbox{ in }\Hf_{\vec{ X}}|(\Path,\bar{a})]\geq \gamma,
\]
we keep decomposing $\Hf_{\vec{ X}}|(\Path,\neg P_1,\bar{a})$.

We then find another maximal set of paths $P_2$ satisfying that 
\[
    \Pr[P_1\mbox{ in }\Hf_{\vec{ X}}|(\Path,\neg P_1,\bar{a})]\geq (1+\delta)^{|P_2|}(\frac{(N-|P_2|-|\Path|)!-|P_1|}{(N-|\Path|)!-|P_1|})^d. 
\]
It is obvious that conditioned on $P_2$, $\Hf_{\vec{ X}}|(\Path,\neg P_1,\bar{a})$ is $(1+\delta)$-uniform following the same calculation. Furthermore, the size of $P_2$ can be bounded by the inequality
\[
    (1+\delta)^{|P_2|} (\frac{(N-|P_2|-|\Path|)!-|P_1|}{(N-|\Path|)!-|P_1|})^d \leq 2^s \frac{1}{\gamma}\cdot(\frac{(N-|P_2|-|\Path|)!-|P_1|}{(N-|\Path|)!-|P_1|})^d.
\]
This implies that $|P_2|\leq \frac{p+\log(1/\gamma)}{\log(1+\delta)}$.

We recursively decompose $\Hf_{\vec{ X}}|(\Path,\bar{a})$ until the probability that the rest is less than $\gamma$.  Then, $\Hf_{\vec{ X}}|(\Path,\bar{a})$ is $\gamma$-close to
\[
    \Hc_{\vec{ X}}|(\bar{a},\Path)=\sum_{i=1}^q \Pr[P_i|\neg P_1,\dots,\neg P_{i-1}]\cdot \Hf_{\vec{ X}}|(\Path,\neg P_1,\dots,\neg P_{i-1},P_i,\bar{a}).
\]
\end{proof}

An algorithm $\A$ which has access to a convex combination of shuffling functions, e.g., $\Hc:= \sum_t p_t \Hf_t$ can be represented as 
\begin{eqnarray*}
    \A^{\Hc} = \sum_t p_t \A^{\Hf_t}. \label{eq:conv}
\end{eqnarray*}  
In the following, we show that if the shuffling is ($p',1-\delta$)-uniform, the probability to find the hidden sets in the shuffling is still bounded as we need for Lemma~\ref{lem:find}.

\begin{claim}\label{lem:rnd3}
Let $p,k,N,N'\in \mathbb{N}$ and $0<\delta<1$. Let $\vec{ X}=(X_1,\dots,X_{k+1})$ be a set of sets with each size $N$. Let $\Hf_{\vec{X}}=(h_1,\dots,h_k)$ be $(p,1+\delta)$-uniform as defined in Def.~\ref{def:almost_uniform} and $\Path$ be the set of $p$ paths fixed in $\Hf_{\vec{X}}$. We choose a $N'$-element set $Y_1$ uniformly randomly from $X_1$, and let $Y_{i} := h(Y_{i-1})$ for $i=2,\dots,k$. Then, for $j\in [k]$, for $x_j\in X_j$, 
\[
    \Pr[x_j\in Y_j]\leq (1+\delta)\cdot\frac{N'}{N-p}. 
\]
\end{claim}
\begin{proof}

It is obvious that $\Pr[x_1\in Y_1] = \frac{N'}{N-p}$ since $Y_1$ is chosen randomly uniformly from $ X_1$.  For $i=2,\dots,k$, $\Pr[x_i\in Y_i]$ can be calculated as follows: 
\begin{eqnarray*}
&&\Pr[x_i\in Y_i] = \Pr[\bigvee_{Y_1\subset X_1} (x_i\in (h_{i-1}\circ\cdots\circ h_1(Y_1))\wedge (Y_1\mbox{ is chosen}))]\\
&&\leq \sum_{Y_1\subset X_1} \Pr[Y_1\mbox{ is chosen}]\cdot \Pr[x_i\in (h_{i-1}\circ\cdots\circ h_1(Y_1))|(Y_1\mbox{ is chosen})]\\
&&\leq \sum_{Y_1\subset X_1} \Pr[Y_1\mbox{ is chosen}]\sum_{y\in Y_1} \Pr[h_{i-1}\circ\cdots\circ h_1(y)= x_i|(Y_1\mbox{ is chosen})]\\
&&\leq (1+\delta) \frac{N'}{N-p}.
\end{eqnarray*}
The first two inequalities follow from the union bound, and the last inequality follows from the fact that $\Hf_{\vec{X}}$ is ($p,1+\delta$)-uniform. 
\end{proof}

\subsection{Proof of Theorem~\ref{thm:bpp_qcd}}

Following the similar idea in previous sections, we want to show that there exists a sequence of shadows which is indistinguishable from $\vecf$. However, to prove Theorem~\ref{thm:bpp_qcd}, we actually show that there exist a ``convex combinatio'' of finitely many shadows, which are indistinguishable from $\vecf$. Specifically, we show that there exist a convex combination $\sum_{t=1}^T p_t(\vecg_t^{(1)},\dots,\vecg_t^{(m)})$ such that
\begin{eqnarray*}
    &&|\Pr[L_m^\vecf\circ\cdots\circ L_1^{\vecf}() = s] - \sum_{t=1}^{T}p_t\cdot\Pr[(L_m^{\vecg_t^{(m)}}\circ\cdots\circ L_1^{\vecg_t^{(1)}}() = s]|\\
    &&\leq md\cdot\sqrt{\frac{\poly(n)}{2^n}}. 
\end{eqnarray*}  
We will give the details of the shadows later in Lemma~\ref{lem:last}. 

We denote a convex combination of bit strings as 
\[
    \bar{z} := \sum_{t=1}^{T} p_t \bar{z}_t, 
\]
where $\bar{z}_1,\dots,\bar{z}_T$ are bit string, $p_1,\dots,p_T$ are the probability that $\bar{z}_{t}$ is sampled, and $T$ is finite in our context.

Let $f$ be a random Simon's function from $\Z_2^n$ to $\Z_2^n$. Let $\vecf$ be the random ($d,f$)-shuffling of $f$. 
In this section, $\bar{s}$ will always be in the form ($\Path,\bar{a}$), where $\Path$ is a set of paths in $\vecf$, and $\bar{a}$ is some bit string correlated to $\vecf$. For example, $\bar{a}$ could be the statement ``$f(0)\oplus f(1) = 1$'' and so on.

We say $\bar{a}$ is uncorrelated to $f$ conditioned on $\Path$ if the procedure producing $\bar{a}$ will not change the output when all mappings in $f_d^*$ except for mappings in $\Path$ are erased by $\bot$.

In the following, we define two kinds of advice, which are ideal and semi-ideal for our analysis.  
\begin{definition}[Ideal advice]\label{def:ideal}
$(\Path,\bar{a})$ is ideal if 
\begin{itemize}
\item $\Path$ does not have a collision in $f$, 
\item $|\Path| = \poly(n)$, 
\item and the bit string $\bar{a}$ is uncorrelated to $f$ conditioned on $\Path$.
\end{itemize}
\end{definition}

\begin{definition}[Semi-ideal advice]\label{def:semi-ideal}
$(\Path,\bar{a})$ is semi-ideal if 
\begin{itemize}
\item $|\Path| = \poly(n)$, 
\item and the bit string $\bar{a}$ is uncorrelated to $f$ conditioned on $\Path$.
\end{itemize}
\end{definition}
Note that the only difference between ideal and semi-ideal advice is that the paths fixed in an ideal advice is promised to have no collision to reveal $s$, while the paths in a semi-ideal advice may have a collision.

\begin{claim}\label{claim:ideal}
Let $\A$ be any algorithm. Let $\bar{s}:=(\Path,\bar{a})$ be an ideal advice of $\vecf$. Then, 
\[
    \Pr[\A(\bar{s})=s] \leq \poly(n)/2^n 
\]
\end{claim}
\begin{proof}
If $\bar{s}$ is ideal, then there is no collision in $\Path$, and $\bar{a}$ is uncorrelated to $f$ conditioned on $\Path$. Given ($\Path, \bar{a}$), $\A$ cannot distinguish whether $f$ is a one-to-one function or a Simon function. Therefore,
\[
    \Pr[\A(\bar{s})=s] \leq \frac{(|\Path|+1)^2}{2^n-(|\Path|+1)^2} \leq \frac{\poly(n)}{2^n}. 
\]
\end{proof}

We consider 
\begin{eqnarray}
\B:= \Pi_{0/1}\circ U_{d+1}\vecf\cdots \vecf U_1\circ \A_{c},\label{eq:B}
\end{eqnarray}
where $U_1,\dots,U_d$ are single depth quantum circuit and $\A_c$ is a PPT algorithm. As we have mentioned earlier, the output of $\B^{\vecf}$ can be represented as ($\Path,\bar{a}$), where $\Path$ is a set of paths fixed by $\A_c$ and $\bar{a}$ is corresponding to the measurement outcome of the quantum circuit.

\begin{lemma}\label{lem:last}
Let $\vecf$ be a ($d,f$)-shuffling sampled from $\D_{f,d}$. Let $\bar{s}:= \sum_{u=1}^{T} p^{(u)} \bar{s}^{(u)}$ be a convex combination of semi-ideal advice strings. For any $\B$ in Eq.~\ref{eq:B}, there exist $\{\convg_{\bar{s}^{(1)}},\dots,\convg_{\bar{s}^{(T)}}\}$, which are convex combinations of sequences of shadows corresponding to $\bar{s}^{(1)},\dots,\bar{s}^{(T)}$ such that for all bit string $\bar{s}'$, for $u\in [T]$, 
\begin{eqnarray*}
&& | \Pr[\B^{\vecf}(\bar{s}^{(u)})=\bar{s}'] - \Pr[\B^{\convg_{\bar{s}^{(u)}}}(\bar{s}^{(u)})=\bar{s}']|\\
&&\leq d\cdot \sqrt{\frac{\poly(n)}{2^n}}, 
\end{eqnarray*}
and the output of $\B^{\convg_{\bar{s}^{(u)}}}(\bar{s}^{(u)})$ is semi-ideal. Moreover, if $\bar{s}$ is ideal, then the output of $\B^{\convg_{\bar{s}^{(u)}}}(\bar{s}^{(u)})$ is also ideal with probability at least $1-\frac{\poly(n)}{2^n}$.
\end{lemma}

Lemma~\ref{lem:last} directly implies that 
\begin{eqnarray*}
&&|\Pr[\B^{\vecf}(\bar{s})=\bar{s}']- \sum_{u=1}^T p^{(u)}\Pr[\B^{\convg_{\bar{s}^{(u)}}}(\bar{s}^{(u)})=\bar{s}']|\\
&&\leq\sum_{u} p^{(u)}|\Pr[\B^{\vecf}(\bar{s}^{(u)})=\bar{s}']- \Pr[\B^{\convg_{\bar{s}^{(u)}}}(\bar{s}^{(u)})=\bar{s}']| \leq d\cdot \sqrt{\frac{\poly(n)}{2^n}}.
\end{eqnarray*}

\begin{proof}[Proof of Lemma~\ref{lem:last}]
For the ease of the analysis, we allow the classical algorithm $\A_c$ to make path queries. We prove the lemma by mathematical induction on the quantum circuit depth.

Given $\bar{s}:=(\Path,\bar{a})$ which is semi-ideal, we let $\Path_0$ be $\Path$ and the set of paths queried by $\A_{c}$, let $\bar{s}^{(0)}:=(\Path_0,\bar{a})$, and let $p:=|\bar{a}|$. Note that the main difficulty which fails the previous proofs is that $\vecf$ may not be uniform conditioned on $\bar{a}$. By applying Lemma~\ref{claim:presamp}, $\vecf|(\bar{s}^{(0)})$ is $\gamma$-close to a convex combination of ($p',1+\delta$)-uniform shuffling functions
\[
    \vecf|(\bar{s}^{(0)}) = \Hc_{\vecS^{(0)}}|(\bar{s}^{(0)}) +\gamma' \Hf',
\]
where $\Hc_{\vecS^{(0)}}|(\bar{s}^{(0)})$ is a convex combination of $(p',1+\delta)$-uniform shuffling functions, $\Hf'$ is an arbitrary random shuffling, and $\gamma'\leq \gamma$. According to claim~\ref{claim:presamp}, $p'\leq \frac{p+\log(1/\gamma)}{\log(1+\delta)}+|\Path_{0}|$. We will set the parameters $p'$, $\gamma$, and $\delta$ shortly. 

We let 
\[
    \Hc_{\vecS^{(0)}}|(\bar{s}^{(0)}):= \sum_{t_1=1}^{T_1} p_{t_1} \cdot \Hf^{(1)}_{t_1}, 
\]
and $\rho^{(0)}$ is the initial state. Then, 
\begin{eqnarray}
    \vecf U_1(\rho^{(0)},\bar{s}^{(0)}) &=& \sum_{t_1=1}^{T_1} p_{t_1} \cdot \Hf^{(1)}_{t_1} U_1(\rho^{(0)},\bar{s}^{(0)}) + \gamma' \Hf'U_1(\rho^{(0)},\bar{s}^{(0)}), \label{eq:conv_2}
\end{eqnarray}
where $\Hf^{(1)}_{1},\dots,\Hf^{(1)}_{T_1}$ are $(p',1+\delta)$-uniform shuffling functions. Moreover, since $\bar{a}$ is uncorrelated to $f$ conditioned on $\Path_0$, the additional paths fixed in $\Hf^{(1)}_{t_1}$ is uncorrelated to $f$ given $\Path_0$. 

Eq.~\ref{eq:conv_2} implies that for all $\bar{z}$
\begin{eqnarray}
&&|\Pr[\Pi_{0/1}\circ U_{d+1}\vecf\cdots\vecf U_1(\rho^{(0)}, \bar{s}^{(0)})=\bar{z}] \nonumber\\
&&- \Pr[\Pi_{0/1}\circ U_{d+1}\vecf\cdots\vecf U_2\left(\Hc_{\vecS^{(0)}}|(\bar{s}^{(0)})\right)U_1(\rho^{(0)}, \bar{s}^{(0)})=\bar{z}]| \leq \gamma,  \label{eq:conv_3}
\end{eqnarray}
where
\begin{eqnarray*}
    &&U_{d+1}\vecf\cdots\vecf U_2\left(\Hc_{\vecS^{(0)}}|(\bar{s}^{(0)})\right) U_1(\rho^{(0)}, \bar{s}^{(0)})\\
    &&=\sum_{t_1=1}^{T_1} p_{t_1}\cdot U_{d+1}\vecf\cdots\vecf U_2\left(\Hf^{(1)}_{t_1}\right) U_1(\rho^{(0)}, \bar{s}^{(0)}).
\end{eqnarray*}

Then, we construct shadows for each $\Hf^{(1)}_{t_1}$ as follows: Let $\Path^{(1)}_{t_1}$ be $\Path_0$ and the set of paths $\Hf^{(1)}_{t_1}$ is fixed on in addition to $\Path_{0}$. We construct the hidden set $\vecS^{(1)}_{t_1}$ based on $\Hf^{(1)}_{t_1}$, $\Path^{(1)}_{t_1}$, and $\vecS^{(0)}$ as in Procedure~\ref{fig:S_3}. Let $\vecg^{(1)}_{t_1}$ be the shadow of $\Hf^{(1)}_{t_1}$ in $\vecS^{(1)}_{t_1}$, Then, 
\begin{eqnarray}
     &&B(\Hf^{(1)}_{t_1} U_1(\rho^{(0)},\bar{s}^{(0)}), \vecg^{(1)}_{t_1} U_1(\rho^{(0)},\bar{s}^{(0)}))\nonumber\\
     &&\leq \sqrt{2\Pr[find\; \vecS^{(1)}_{t_1}: U_1^{\Hf^{(1)}_{t_1}\setminus \vecS^{(1)}_{t_1}},\rho^{(0)}]}\nonumber\\
     &&\leq \sqrt{(1+\delta) \frac{2q_1}{2^n}}\label{eq:conv_4}
\end{eqnarray}
where $q_1$ is the number of queries $U_1$ performs. The first inequality follows from Lemma~\ref{lem:o2h} and the last inequality follows from Claim~\ref{lem:rnd3} and Lemma~\ref{lem:find}. The output of $\vecg^{(1)}_{t_1} U_1(\rho^{(0)}, \bar{s}^{(0)})$ is uncorrelated to $\vecf$ in $\vecS^{(1)}_{t_1}$ by following the definition of $\vecg^{(1)}_{t_1}$. 
\floatname{algorithm}{Procedure}
\begin{algorithm}[h]
    \begin{mdframed}[style=figstyle,innerleftmargin=10pt,innerrightmargin=10pt]
    Given $j\in \mathbb{N}$, $\vecS^{(j-1)}:= (S^{(j-1)}_{j-1},\dots,S^{(j-1)}_{d})$, $\Hf_{k,j}= (h_{j},\dots,h_{d})$, and $\Path$ a set of fixed paths.
    \begin{enumerate}
        \item Let $S_j^{(j-1)}$ be $S_j^{(j-1)}$ except for elements on $\Path$.  
        \item Let $S_j^{(j)}$ be a subset chosen uniformly at random with the promise that $|S_j^{(j)}|/|S_j^{(j-1)}| =\frac{1}{2^n}$, and $S_j^{(j)}$ includes all elements in $S_j$ except for elements on $\Path$. 
        \item For $\ell=j+1,\dots, d$,  let $S_{\ell}^{(j)}:= \{h_{\ell-1}\circ\cdots\circ h_{j}(S_{j}^{(j)})\}$. 
        \item We let $\vecS^{(j)} = (S_{j}^{(j)},\dots,S_{d}^{(j)})$.
    \end{enumerate}
    \caption{The hidden sets for $d$-CQ scheme}
    \label{fig:S_3}
    \end{mdframed}
\end{algorithm}

By combining Eq.~\ref{eq:conv_3} and Eq.~\ref{eq:conv_4}, we have proven that for all $\bar{z}$
\begin{eqnarray*}
&&|\Pr[\Pi_{0/1}\circ U_{d+1}\vecf\cdots\vecf U_1(\rho^{(0)}, \bar{s}^{(0)})=\bar{z}] \\
&&- \sum_{t_1=1}^{T_1}p_{t_1}\cdot \Pr[\Pi_{0/1}\circ U_{d+1}\vecf\cdots\vecf U_2\vecg^{(1)}_{t_1} U_1(\rho^{(0)}, \bar{s}^{(0)})=\bar{z}]| \\
&&\leq \gamma+ \sqrt{(1+\delta) \frac{q_1}{2^n}},  \label{eq:conv_5}
\end{eqnarray*}
The output state of $\sum_{t_1=1}^{T_1} p_{t_1}\cdot \vecg^{(1)}_{t_1} U_1(\rho^{(0)}, \bar{s}^{(0)})$ is
\begin{eqnarray*}
  \rho^{(1)}&:=&\sum_{t_1=1}^{T_1}p_{t_1}\cdot \vecg^{(1)}_{t_1} U_1 (\rho^{(0)}, \bar{s}^{(0)}) \\
  &:=& \sum_{t_1=1}^{T_1}p_{t_1}\cdot \rho^{(1)}_{t_1}, 
\end{eqnarray*}
and we let $\bar{s}^{(1)}:= \sum_{t_1=1}^{T_1} p_{t_1} \bar{s}^{(1)}_{t_1}$, where $\bar{s}^{(1)}_{t_1} := (\Path^{(1)}_{t_1},\bar{a})$.

$\bar{s}^{(1)}_{t_1}$ must be semi-ideal. First, $\bar{a}$ is still uncorrelated to $f$ conditioned on $\Path^{(1)}_{t_1}$ because $\Path\subseteq \Path^{(1)}_{t_1}$. By Claim~\ref{claim:presamp}, $p'\leq \frac{p'+\log 1/\gamma}{\log (1+\delta)}$. The size $|\Path^{(1)}_{t_1}|= |\Path_0| + p'$ is at most $\poly(n)$ by setting $\gamma=1/\poly(n)$ and $\delta = O(1)$.

We then show that if $\bar{s}$ is ideal, $\bar{s}^{(1)}_{t_1}$ for $t_1\in [T_1]$ is ideal with high probability at least $1-\frac{\poly(n)}{2^n}$. Note that $\bar{a}$ is uncorrelated to $f$ conditioned on $\Path$. This implies that for all $\Hf^{(1)}_{t_1}\in \Hc_{\vecS^{(0)}}|\bar{s}^{(0)}$,  $\Hf^{(1)}_{t_1}$ can only have at most additional $p'$ paths be fixed in $\vecf$, and these paths are uncorrelated to $f$ conditioned on $\Path_0$. Since these paths are uncorrelated to $f$ conditioned on $\Path_0$, the probability that $\Path^{(1)}_{t_1}$ gives $s$ is at most $\frac{(|\Path^{(1)}_{t_1}|+1)^2}{2^n-(|\Path^{(1)}_{t_1}|+1)^2}$. Therefore, the probability that $\bar{s}^{(1)}$ is ideal is at least $1-\frac{\poly(n)}{2^n}$.

Then, we consider 
\[
    \Pi_{0/1}\circ U_{d+1}\vecf\cdots\vecf U_2(\rho^{(1)}, \bar{s}^{(1)}). 
\]
The formula above can be decomposed as 
\begin{eqnarray*}
    \sum_{t_1=1}^{T_1} p_{t_1}\cdot \Pi_{0/1}\circ U_{d+1}\vecf\cdots\vecf U_2(\rho^{(1)}_{t_1}, \bar{s}^{(1)}_{t_1}). 
\end{eqnarray*}
Note that $\rho^{(1)}_{t_1}$ is uncorrelated to mappings in $\vecS^{(1)}_{t_1}$ since $\vecg^{(1)}_{t_1}$ has mappings in $\vecS^{(1)}_{t_1}$ be blocked. This implies that conditioned on $\rho^{(1)}_{t_1}$, the mappings in $\vecS^{(1)}_{t_1}$ are still uniformly random. Therefore, for each input ($\rho^{(1)}_{t_1},\bar{s}^{(1)}_{t_1}$), we can apply the presampling argument in Claim~\ref{claim:presamp} again and get the convex combination $\Hc_{\vecS^{(1)}_{t_1}}|(\bar{s}^{(1)}_{t_1})$ satisfying that for all $\bar{z}$, 
\begin{eqnarray}
&&|\Pr[\Pi_{0/1}\circ U_{d+1}\circ U_d\cdots\vecf U_2(\rho^{(1)}_{t_1}, \bar{s}^{(1)}_{t_1})=\bar{z}] \nonumber\\
&&- \Pr[\Pi_{0/1}\circ U_{d+1}\vecf\cdots\vecf U_2\left(\Hc_{\vecS^{(1)}_{t_1}}|(\bar{s}^{(1)}_{t_1})\right) U_1(\rho^{(1)}_{t_1}, \bar{s}^{(1)}_{t_1})=\bar{z}]|\leq \gamma.\label{eq:conv_6}
\end{eqnarray}

Here, we represent the convex combination as following formula 
\[
    \Hc_{\vecS^{(1)}_{t_1}}|(\bar{s}^{(1)}_{t_1}):= \sum_{t_2=1}^{T_2} p_{t_1,t_2} \Hf^{(2)}_{t_1,t_2} 
\]
and then construct the shadow $\vecg^{(2)}_{t_1,t_2}$ and the hidden set $\vecS^{(2)}_{t_1,t_2}$ for each $\Hf^{(2)}_{t_1,t_2}$ according to Procedure~\ref{fig:S_3}. It satisfies that 
\begin{eqnarray}
     &&B(\Hf^{(2)}_{t_1,t_2} U_2(\rho^{(1)}_{t_1}, \bar{s}^{(1)}_{t_1}),\, \vecg^{(2)}_{t_1,t_2} U_2(\rho^{(1)}_{t_1}, \bar{s}^{(1)}_{t_1}))\nonumber\\
     &&\leq \sqrt{2\Pr[find\; \vecS^{(2)}_{t_1,t_2}: U_1^{\Hf^{(2)}_{t_1,t_2}\setminus \vecS^{(2)}_{t_1,t_2}},\rho^{(0)}]}\nonumber\\
     &&\leq \sqrt{(1+\delta) \frac{2q_2}{2^n}},\label{eq:conv_7}
\end{eqnarray}
where $q_2$ is the number of queries $U_2$ performs. We let $\Path^{(2)}_{t_1,t_2}$ be $\Path^{(1)}_{t_1}$ and the additional set of paths fixed in $\Hf^{(2)}_{t_1,t_2}$. We let $\bar{s}^{(2)}_{t_1,t_2}:= (\bar{a}, \Path^{(2)}_{t_1,t_2})$.

By Eq.~\ref{eq:conv_6} and Eq.~\ref{eq:conv_7}, for all $\bar{z}$, 
\begin{eqnarray*}
    &&|\sum_{t_1=1}^{T_1} p_{t_1}\cdot\Pr[\Pi_{0/1}\circ U_{d+1}\vecf\circ U_d\circ\cdots\circ\vecf\circ U_2(\rho^{(1)}_{t_1}, \bar{s}^{(1)}_{t_1})=\bar{z}] -\\
    && \sum_{t_1=1}^{T_2}\sum_{t_2=1}^{T_2} p_{t_1}p_{t_1,t_2} \cdot \Pr[\Pi_{0/1}\circ U_{d+1}\vecf\cdots\vecg^{(2)}_{t_1,t_2} U_2(\rho^{(1)}_{t_1}, \bar{s}^{(1)}_{t_1})=\bar{z}]|\\
    &&\leq \gamma + \sqrt{(1+\delta) \frac{q_2}{2^n}}. 
\end{eqnarray*}

Again, we let the output state of $\sum_{t_1=1}^{T_1}\sum_{t_2=1}^{T_2} p_{t_1}p_{t_1,t_2} \cdot \vecg^{(2)}_{t_1,t_2} U_2 \vecg^{(1)}_{t_1} U_1$ be
\begin{eqnarray*}
    \rho^{(2)} &:=& \sum_{t_1=1}^{T_1}\sum_{t_2=1}^{T_2} p_{t_1}p_{t_1,t_2} \vecg^{(2)}_{t_1,t_2} U_2 \vecg^{(1)}_{t_1} U_1(\rho^{(0)}, \bar{s}^{(0)})\\
    &:=& \sum_{t_1=1}^{T_1}\sum_{t_2=1}^{T_2} p_{t_1} p_{t_1,t_2}\rho^{(2)}_{t_1,t_2}, 
\end{eqnarray*}
which satisfies that $\rho^{(2)}_{t_1,t_2}$ is uncorrelated the mappings in $\vecS^{(2)}_{t_1,t_2}$. We let
\begin{eqnarray*}
    \bar{s}^{(2)}:= \sum_{t_1=1}^{T_1}\sum_{t_2=1}^{T_2} p_{t_1}p_{t_1,t_2} \bar{s}^{(2)}_{t_1,t_2}.    
\end{eqnarray*}
Here, in case that $\bar{s}^{(1)}_{t_1}$ is ideal,  $\bar{s}^{(2)}_{t_1,t_2}$ is also ideal with probability at least $1-\frac{\poly(n)}{2^n}$ via the same analysis. Therefore, the probability that $\bar{s}^{(2)}_{t_1,t_2}$ is ideal is at least $1-\frac{\poly(n)}{2^n}$. $\bar{s}^{(2)}_{t_1,t_2}$ must be semi-ideal since $\bar{s}^{(1)}_{t_1}$ is semi-ideal via the same analysis.

Now, we can suppose when the $k$-th parallel queries are applied, for all $\bar{z}$, there exist $\{\vecg^{(k)}_{t_1,\dots,t_k}\}$ and the corresponding hidden sets $\{\vecS^{(k)}_{t_1,\dots,t_k}\}$ such that
\begin{eqnarray}
    &&|\sum_{t_1,\dots,t_{k-1}} p_{t_1}\cdots p_{t_1,\dots,t_{k-1}}\cdot\Pr[\Pi_{0/1}\circ U_{d+1}\vecf\cdots\vecf U_k(\rho^{(k-1)}_{t_1,\dots,t_{k-1}}, \bar{s}^{(k-1)}_{t_1,\dots,t_{k-1}})=\bar{z}]-\nonumber \\
    && \sum_{t_1,\dots,t_{k}} p_{t_1}\cdots p_{t_1,\dots,t_{k}} \Pr[\Pi_{0/1}\circ U_{d+1}\vecf\cdots\vecg^{(k)}_{t_1,\dots,t_k} U_k(\rho^{(k-1)}_{t_1,\dots,t_{k-1}}, \bar{s}^{(k-1)}_{t_1,\dots,t_{k-1}})=\bar{z}]|\nonumber\\
    &&\leq \gamma + \sqrt{(1+\delta) \frac{q_k}{2^n}}. \label{eq:conv_10}
\end{eqnarray}

Here, $\rho^{(k-1)}_{t_1,\dots,t_{k-1}}$ is the output of $\vecg^{(k-1)}_{t_1,\dots,t_{k-1}}U_{k-1}\cdots \vecg^{(1)}_{i,t_1}U_1 (\rho^{(0)}, \bar{s}^{(0)})$ and 
\[
\bar{s}^{(k-1)}_{t_1,\dots,t_{k-1}}:= (\bar{a}, \Path^{(k-1)}_{t_1,\dots,t_{k-1}}),
\]
where $\bar{s}^{(k-1)}_{t_1,\dots,t_{k-1}}$ is ideal with probability at least $1-\frac{\poly(n)}{2^n}$, and $\rho^{(k-1)}_{t_1,\dots,t_{k-1}}$ is uncorrelated to mappings in $\vecS^{(k-1)}_{t_1,\dots,t_{k-1}}$. 

The output of the scheme with access to $\{\vecg^{(k)}_{t_1,\dots,t_k}\}$ in Eq.~\ref{eq:conv_10} is
\begin{eqnarray*}
\rho^{(k)} &:=& \sum_{t_1,\dots,t_k} p_{t_1}\cdots p_{t_1,\dots,t_k} \cdot \rho^{(k)}_{t_1,\dots,t_{k}}
\end{eqnarray*}
which satisfies that $\rho^{(k)}_{t_1,\dots,t_{k}}$ is uncorrelated to the mappings in $\vecS^{(k)}_{t_1,\dots,t_{k}}$. We let 
\[
\bar{s}^{(k)}:=\sum_{t_1,\dots,t_k} (p_{t_1}\cdots p_{t_1,\dots,t_k} \bar{s}^{(k)}_{t_1,\dots,t_{k}},
\] 
where  $\bar{s}^{(k)}_{t_1,\dots,t_{k}}:= \Path^{(k)}_{t_1,\dots,t_{k}},\bar{a})$ is ideal with probability at least $1-\frac{\poly(n)}{2^n}$.

Consider the ($k+1$)-th quantum parallel queries,  
\begin{eqnarray*}
&&\Pi_{0/1}\circ U_{d+1}\vecf\cdots\vecf U_{k+1}(\rho^{(k)}, \bar{s}^{(k)})\\
&&= \sum_{t_1,\dots,t_{k}} p_{t_1}\cdots p_{t_1,\dots,t_{k}}\cdot\Pi_{0/1}\circ U_{d+1}\vecf\cdots\vecf U_k(\rho^{(k)}_{t_1,\dots,t_{k}}, \bar{s}^{(k)}_{t_1,\dots,t_{k}}).
\end{eqnarray*}

Following the fact that $\rho^{(k)}_{t_1,\dots,t_{k}}$ is only correlated to mappings out of $\vecS^{(k)}_{t_1,\dots,t_{k}}$, we apply the presampling argument in Claim~\ref{claim:presamp}, and there exists a convex combination
\[
\Hc_{\vecS^{(k)}_{t_1,\dots,t_{k}}}|(\bar{s}^{(k)}_{t_1,\dots,t_{k}}):= \sum_{t_{k+1}} p_{t_1,\dots,t_{k+1}} \Hf^{(k+1)}_{t_1,\dots,t_{k+1}}
\]
such that for all $\bar{z}$, 
\begin{eqnarray*}
    &&|\sum_{t_1,\dots,t_{k}} p_{t_1}\cdots p_{t_1,\dots,t_{k}}\cdot\Pr[\Pi_{0/1}\circ U_{d+1}\vecf\cdots\vecf U_{k+1}(\rho^{(k)}_{t_1,\dots,t_{k}}, \bar{s}^{(k)}_{t_1,\dots,t_{k}})=\bar{z}]-\nonumber \\
    &&\sum_{t_1,\dots,t_{k+1}} p_{t_1}\cdots p_{t_1,\dots,t_{k+1}} \Pr[\Pi_{0/1}\circ U_{d+1}\vecf\cdots\Hf^{(k+1)}_{t_1,\dots,t_{k+1}} U_{k+1}(\rho^{(k)}_{t_1,\dots,t_{k}}, \bar{s}^{(k)}_{t_1,\dots,t_{k}})=\bar{z}]|\nonumber\\
    &&\leq \gamma. 
\end{eqnarray*}

We then construct $\vecg^{(k+1)}_{t_1,\dots,t_{k+1}}$ and $\vecS^{(k+1)}_{t_1,\dots,t_{k+1}}$ for each $\Hf^{(k+1)}_{t_1,\dots,t_{k+1}}$. Following the same proof, we can show that for all $\bar{z}$
\begin{eqnarray*}
    &&|\sum_{t_1,\dots,t_{k}} p_{t_1}\cdots p_{t_1,\dots,t_{k}}\cdot\Pr[\Pi_{0/1}\circ U_{d+1}\vecf\cdots\vecf U_{k+1}(\rho^{(k)}_{t_1,\dots,t_{k}}, \bar{s}^{(k)}_{t_1,\dots,t_{k}})=\bar{z}]-\nonumber \\
    && \sum_{t_1,\dots,t_{k+1}} p_{t_1}\cdots p_{t_1,\dots,t_{k+1}} \Pr[\Pi_{0/1}\circ U_{d+1}\vecf\cdots\vecg^{(k+1)}_{t_1,\dots,t_{k+1}} U_{k+1}(\rho^{(k+1)}_{t_1,\dots,t_{k+1}}, \bar{s}^{(k+1)}_{t_1,\dots,t_{k+1}})=\bar{z}]|\nonumber\\
    &&\leq \gamma + \sqrt{(1+\delta) \frac{q_{k+1}}{2^n}}.  
\end{eqnarray*}
Moreover, the output of the scheme with access to $\{\vecg^{(k+1)}_{t_1,\dots,t_{k+1}}\}$ is
\begin{eqnarray*}
\rho^{(k+1)} &:=& \sum_{t_1,\dots,t_{k+1}} p_{t_1}\cdots p_{t_1,\dots,t_{k+1}} \cdot \rho^{(k+1)}_{t_1,\dots,t_{k+1}}
\end{eqnarray*}
which satisfies that $\rho^{(k+1)}_{t_1,\dots,t_{k+1}}$ is uncorrelated to the mappings in $\vecS^{(k+1)}_{t_1,\dots,t_{k+1}}$. We let 
\[
\bar{s}^{(k+1)}:=\sum_{t_1,\dots,t_{k+1}} (p_{t_1}\cdots p_{t_1,\dots,t_{k+1}} \bar{s}^{(k+1)}_{t_1,\dots,t_{k+1}}),
\] 
where  $\bar{s}^{(k+1)}_{t_1,\dots,t_{k+1}}:= (\Path^{(k+1)}_{t_1,\dots,t_{k+1}},\bar{a})$ is ideal with probability $1-\frac{\poly(n)}{2^n}$. Then, for all $\bar{z}$
\begin{eqnarray}
&&|\Pr[\B^{\vecf}(\bar{s}^{(u)})=\bar{z}] - \Pr[\B^{\convg_{\bar{s}^{(u)}}}(\bar{s}^{(u)})=\bar{z}]|\nonumber\\
&&= |\Pr[\Pi_{0/1}\circ U_{d+1}\vecf\cdots\vecf U_1(\rho^{(0)}, \bar{s}^{(0)})=\bar{z}]-\nonumber \\
&&\quad\sum_{t_1,\dots,t_{d}} p_{t_1}\cdots p_{t_1,\dots,t_{d}} \Pr[\Pi_{0/1}\circ U_{d+1}\vecg^{(d)}_{t_1,\dots,t_d}\cdots\vecg^{(1)}_{t_1} U_1(\rho^{(0)}, \bar{s}^{(0)})=\bar{z}]|\nonumber\\
&&\leq d\gamma+ \sum_{i=1}^d\sqrt{(1+\delta) \frac{q_i}{2^n}}.  \label{eq:conv_8}
\end{eqnarray}
Eq.~\ref{eq:conv_8} follows from the hybrid argument and the indistinguishability we have just proven by math induction.

Finally, we need to show that the output of $\B^{\convg_{\bar{s}^{(u)}}}(\bar{s}^{(u)})$ is ideal with probability $1-\frac{\poly(n)}{2^n}$. The output can be represented by 
\[
    \B^{\convg_{\bar{s}^{(u)}}}(\bar{s}^{(u)}) := \sum_{t_1,\dots,t_d} p_{t_1}\cdots p_{t_1,\dots,t_d} \bar{s}^{(d)}_{t_1,\dots,t_d} 
\]
and
\[
    \bar{s}^{(d)}_{t_1,\dots,t_d}:= (\bar{a}',\Path^{(d)}_{t_1,\dots,t_d})).
\]

First, we show that $|\Path^{(d)}_{t_1,\dots,t_d}|$ is at most a polynomial in $n$. The number of fixed paths by the sequence of shadows is at most $d\cdot \frac{p'+\log 1/\gamma}{\log 1+\delta}$. We set $\gamma = 1/\poly(n)$ and $\delta=O(1)$ such that $d\gamma+ \sum_{i=1}^d\sqrt{(1+\delta) \frac{q_i}{2^n}}\leq d\sqrt{\frac{\poly(n)}{2^n}}$ and $|\Path^{(d)}_{t_1,\dots,t_d}|$ is at most $\poly(n)$. Moreover, since $\bar{a}$ is uncorrelated to $f$ given $\Path$, the additional paths fixed in $\vecg^{(1)}_{t_1},\dots,\vecg^{(d)}_{t_1,\dots,t_d}$ is also uncorrelated to $f$ given $\Path_0$. This implies that the probability that $\Path^{(d)}_{t_1,\dots,t_d}$ gives the hidden shift $s$ is at most $\frac{(|\Path^{(d)}_{t_1,\dots,t_d}|+1)^2}{2^n-(|\Path^{(d)}_{t_1,\dots,t_d}|+1)^2}$. In the last, the bit string $\bar{a}'$ is from the quantum state $U_{d+1}\rho^{(d)}U^{\dag}_{d+1}$ and $\bar{a}$. Since $\rho^{(d)}$ is uncorrelated to $f_d^*$ conditioned on $\Path^{(d)}_{t_1,\dots,t_d}$, $\bar{a}'$ must be uncorrelated to $f$ conditioned on $\Path^{(d)}_{t_1,\dots,t_d}$. 
\end{proof}

Now, we are ready to prove Theorem~\ref{thm:bpp_qcd}. 
\begin{proof}[Proof of Theorem~\ref{thm:bpp_qcd}]

We first consider $L_1$. Since the input of $L_1$ does not have any short advice which is correlated to all $\vecf$, we can use the analysis in Theorem~\ref{thm:qcd_bpp} to show that there exists a sequence of shadow $\vecg^{(1)}$ such that 

\[
    |\Pr[(L_m\circ\cdots\circ L_1)^{\vecf}()=s] - \Pr[(L_m\circ\cdots\circ L_2)^{\vecf}\circ L_1^{\vecg^{(1)}}()=s]\leq d\sqrt{\frac{\poly(n)}{2^n}},   
\]
and the output of $L_1^{\vecg^{(1)}}$ is ideal with probability $1-\frac{\poly(n)}{2^n}$ and must be semi-ideal. Or, alternatively, we can follow Lemma~\ref{lem:last}, which also gives the same conclusion. 

Then, we apply Lemma~\ref{lem:last} to replace the oracle of $L_2$, 
\[
    |\Pr[(L_m\circ\cdots\circ L_2)^{\vecf}(L_1^{\vecg^{(1)}}())=s] - \Pr[(L_m\circ\cdots\circ L_3)^{\vecf}\circ L_2^{\convg^{(2)}}(L_1^{\vecg^{(1)}}())=s]\leq d\sqrt{\frac{\poly(n)}{2^n}}. 
\]
We continuously replace the oracle of $L_i$ for $i=1,\dots,m$ according to Lemma~\ref{lem:last}. Finally, we can get the inequality
\begin{eqnarray}
    |\Pr[(L_m\circ\cdots\circ L_1)^{\vecf}()=s] - \Pr[L_m^{\convg^{(m)}}\circ\cdots\circ L_1^{\vecg^{(1)}}()=s]\leq dm\sqrt{\frac{\poly(n)}{2^n}}.\label{eq:last_1}
\end{eqnarray}

Then, the rest to show is that $\Pr[(L_m^{\convg^{(m)}}\circ\cdots\circ L_1^{\vecg^{(1)}}()=s]$ is negligible. By Lemma~\ref{lem:last}, the output of $L_m^{\convg^{(m)}}\circ\cdots\circ L_1^{\vecg^{(1)}}()$ is ideal with probability at least $1-\frac{\poly(n)}{2^n}$. Then,  by applying Claim~\ref{claim:ideal}, 
\begin{eqnarray}
\Pr[\A_c\left(L_m^{\convg^{(m)}}\circ\cdots\circ L_1^{\vecg^{(1)}}()\right)=s] \leq \frac{\poly(n)}{2^n}. \label{eq:last_2}
\end{eqnarray}

Finally, by combining Eq.~\ref{eq:last_1} and Eq.~\ref{eq:last_2}, 
\[
    \Pr[(L_m\circ\cdots\circ L_1)^{\vecf}()=s] \leq 
    dm\sqrt{\frac{\poly(n)}{2^n}} + \frac{\poly(n)}{2^n}. 
\]
This completes the proof. 
\end{proof}

\subsection{On separating the depth hierarchy of \texorpdfstring{$d$}{Lg}-CQ scheme}
By using the same proof for Theorem~\ref{thm:bpp_qcd}, we can show that the $\SSP{d}$ is also hard for any $d$-CQ scheme.  
\begin{theorem}\label{thm:bpp_qcd_2}
The $\SSP{d}$ cannot be decided by any $d$-CQ scheme with probability greater than $\frac{1}{2}+\sqrt{\frac{\poly(n)}{2^n}}$.  
\end{theorem}
\begin{proof}
Following the proof for Theorem~\ref{thm:bpp_qcd}, there exist $\vecg^{(1)},\convg^{(2)},\dots,\convg^{(m)}$ such that $\A$ cannot distinguish $\vecf$ from $\vecg^{(1)},\convg^{(2)},\dots,\convg^{(m)}$. Moreover, in case that $f$ is a random Simon function, $\A$ with access to $\vecg^{(1)},\dots,\convg^{(m)}$ cannot find $s$. Therefore, $\Pr[\A^{\vecf}()=1]\leq 1/2+md\cdot\sqrt{\frac{\poly(n)}{2^n}}$ 
\end{proof}

\begin{corollary}\label{cor:bpp_qnc_2}
For any $d\in \mathbb{N}$, there is a $(2d+1)$-CQ scheme which can solve the $\SSP{d}$ with high probability, but there is no $d$-CQ scheme which can solve the $\SSP{d}$. 
\end{corollary}
\begin{proof}
This corollary follows from Theorem~\ref{thm:bpp_qcd_2} and Theorem~\ref{thm:dSSP_solve} directly. 
\end{proof}

Finally, we can conclude that  
\begin{corollary}\label{cor:bpp_qcd_3}
Let $\Ora$ and $\Lang(\Ora)$ be defined as in Def.~\ref{def:language}. $\Lang(\Ora)\in \class{BQP}^{\Ora}$ and $\Lang(\Ora)\notin \class{(BPP^{BQNC})}^{\Ora}$. 
\end{corollary}
\begin{proof}
Note that each $n\in \mathbb{N}$,  $\Ora_{unif}^{f_n,d(n)}\in \Ora$ has depth equal to the input size. A quantum circuit with depth $\poly(n)$ can decide whether $f_n$ is a Simon's function by  Theorem~\ref{thm:dSSP_solve} and thus decides if $1^n$ is in $\Lang(\Ora)$. However, for $d$-CQ scheme which only has quantum depth $d = \poly\log n$, it cannot decide the language by Theorem~\ref{thm:bpp_qcd_2}.
\end{proof}